\documentclass[a4paper,10pt]{article}

\usepackage[margin=1.25in]{geometry}
\usepackage{microtype}
\usepackage{cite}    % sorts citations by numbers

\usepackage{tikz}
\usepackage{pgffor}
\usepackage[utf8]{inputenc}
\usepackage{framed}
\usepackage{boxedminipage}
\usepackage{graphicx}
\usepackage{amsmath,amssymb,amsthm}
\usepackage[textsize=tiny,textwidth=2.2cm,shadow]{todonotes}
\usepackage{caption}
\usepackage{subcaption}
\usepackage{thmtools}
\usepackage{thm-restate}
\usepackage{authblk}

\usetikzlibrary{patterns}

\newtheorem{theorem}{Theorem}

\newtheorem{corollary}[theorem]{Corollary}
\newtheorem{lemma}[theorem]{Lemma}

\newtheorem{algorithm}[theorem]{Algorithm}
\newtheorem{condition}[theorem]{Condition}

\newcommand{\e}{\varepsilon}

\newcommand{\OPT}{\textsc{Opt}}

\newcommand{\supp}[1]{\mathop{\textnormal{supp}}(#1)}

\newcommand{\jcom}[1]{}
\newcommand{\kkcom}[1]{}
\newcommand{\kjcom}[1]{}

\newcommand{\gz}{{\mathbb{Z}}}

\newcommand{\Jbig}{\mathcal{J}_{\textnormal{big}}}
\newcommand{\Jsmall}{\mathcal{J}_{\textnormal{small}}}
\newcommand{\Jround}{\mathcal{J}^*}
\newcommand{\Jnew}{\mathcal{J}_{\textnormal{new}}}

\renewcommand{\textcolor}[2]{#2}

\title{Closing the Gap for Makespan Scheduling\\ via Sparsification Techniques\footnote{This work was partially supported by DFG Project, Entwicklung und Analyse von effizienten polynomiellen Approximationsschemata f\"ur Scheduling- und verwandte Optimierungsprobleme, Ja 612/14-2, by FONDECYT project 3130407, and by Nucleo Milenio Informaci\'on y Coordinaci\'on en Redes ICM/FIC RC130003.}}

\author[1]{Klaus Jansen}
\author[1]{Kim-Manuel Klein}
\author[2]{Jos\'e Verschae}

\affil[1]{University of Kiel, Department of Computer Science, Kiel, Germany. E-mail:\texttt{\{kj,kmk\}@informatik.uni-kiel.de} }
\affil[2]{Pontificia Universidad Católica de Chile, Facultad de Matemáticas \& Escuela de Ingeniería, Santiago, Chile. E-mail:\texttt{jverschae@uc.cl}}

\date{\today}

\begin{document}

\maketitle

\begin{abstract}
Makespan scheduling on identical machines is one of the most basic and fundamental packing problems studied in the discrete optimization literature. It asks for an assignment of $n$ jobs to a set of $m$ identical machines that minimizes the makespan. The problem is strongly NP-hard, and thus we do not expect a $(1+\e)$-approximation algorithm with a running time that depends polynomially on $1/\e$. Furthermore, Chen et al.~\cite{chen_optimality_2013} recently showed that a running time of $2^{(1/\e)^{1-\delta}}+\text{poly}(n)$ for any $\delta>0$ would imply that the Exponential Time Hypothesis (ETH) fails. 
A long sequence of algorithms have been developed that try to obtain low dependencies on $1/\e$, the better of which achieves a running time of $2^{\tilde{O}(1/\e^2)}+O(n\log n)$~\cite{jansen_eptas_2010}. In this paper we obtain an algorithm with a running time of $2^{\tilde{O}(1/\e)}+O(n\log n)$, which is tight under ETH up to logarithmic factors on the exponent. 

Our main technical contribution is a new structural result on the \emph{configuration-IP}. More precisely, we show the existence of a highly symmetric and sparse optimal solution, in which all but a constant number of machines are assigned a configuration with small support. This structure can then be exploited by integer programming techniques and enumeration. We believe that our structural result is of independent interest and should find applications to other settings. In particular, we show how the structure can be applied to the minimum makespan problem on related machines and to a larger class of objective functions on parallel machines. For all these cases we obtain an efficient PTAS with running time \textcolor{red}{$2^{\tilde{O}(1/\e)} + \text{poly}(n)$.}
 \end{abstract}

\section{Introduction}
\label{sec:intro}

Minimum makespan scheduling is one of the foundational problems in the literature on approximation algorithms~\cite{graham_bounds_1966, graham_bounds_1969}.
In the \emph{identical machine} setting the problem asks for an assignment of a set of $n$ jobs $\mathcal{J}$ to a set of $m$ identical machines $\mathcal{M}$. Each job $j\in\mathcal{J}$ is characterized by a non-negative processing time $p_j \in \mathbb{Z}_{> 0}$. The load of a machine is the total processing time of jobs assigned to it, and our objective is to minimize the \emph{makespan}, that is, the maximum machine load. This problem is usually denoted $P||C_{\max}$. It is well known to admit a \emph{polynomial time approximation scheme} (PTAS)~\cite{hochbaum_using_1987}, and there has been many subsequent works improving the running time or deriving PTAS's for more general settings. The fastest PTAS for $P||C_{\max}$ achieves a running time of $2^{O(1/\e^2)\log^3(1/\e))}+O(n\log n)$ for $(1+\e)$-approximate solutions~\textcolor{red}{\cite{jansen_eptas_2010}}.  Very recently, Chen et al.~\cite{chen_optimality_2013} showed that, assuming the \emph{exponential time hypothesis} (ETH), there is no PTAS that yields $(1+\e)$-approximate solutions for $\e>0$ with running time $2^{(1/\e)^{1-\delta}}+ \text{poly}(n)$ for any $\delta>0$~\cite{chen_optimality_2013}.

Given a guess $T\in \mathbb{N}$ on the optimal makespan, which can be found with binary search, the problem reduces to deciding the existence of a packing of the jobs to $m$ machines (or bins) of capacity $T$. If we aim for a $(1+\e)$-approximate solution, for some $\e>0$, we can assume that all processing times are integral and $T$ is a constant number, namely $T \in O(1/\e^2)$. This can be achieved with well known rounding and scaling techniques~\cite{alon_approximation_1997, alon_approximation_1998, hochbaum_approximation_1997} which will be  specified later. Let $\pi_1< \pi_2<\ldots<\pi_d$ be the job sizes appearing in the instance after rounding, and let $b_k$ denote the number of jobs of size $\pi_k$. The mentioned rounding procedure implies that the number of different job sizes is $d=O((1/\e)\log (1/\e))$. Hence, for large $n$ we obtain a highly symmetric problem where several jobs will have the same processing time. Consider the \emph{knapsack polytope} $\mathcal{P} = \{c\in \mathbb{R}^{d}_{\geq 0}: \pi \cdot c \le T\}$. A packing on one machine can be expressed as a vector $c\in Q=\mathbb{Z}^d\cap \mathcal{P}$, where $c_k$ denotes the number of jobs of size $\pi_k$ assigned to the machine. Elements in $Q=\mathbb{Z}^d\cap \mathcal{P}$ are called \emph{configurations}. Considering a variable $x_c\in\mathbb{Z}_{\ge0}$ that decides the multiplicity of configuration $c$ in the solution, our problem reduces to solving the following linear integer program (ILP):

\begin{align}
 \label{eq:ConicComb}
 \text{[conf-IP]}\quad  & \sum_{c\in Q} c\cdot x_c =b,\\
 \label{eq:NumberConf}
 & \sum_{c\in Q} x_c = m,\\
 &x_c \in \mathbb{Z}_{\ge0} & \text{ for all }c\in Q. 
\end{align}

In this article we derive new insights on this ILP that help us \textcolor{red}{to design} faster algorithms for $P||C_{\max}$ and other more general problems. These including makespan scheduling on \emph{related machines} $Q||C_{\max}$, and a more general class of objective functions on parallel machines. We show that all these problems admit a PTAS with running time $2^{O((1/\e)\log^4(1/\e))} + \text{poly}(n)$. Hence, our algorithm is best possible up to polylogarithmic factors in the exponent assuming ETH~\cite{chen_optimality_2013}. 

\subsection{Literature Review}

There is an old chain of approximation algorithms for $P||C_{\max}$, starting from the seminal work by Graham~\cite{graham_bounds_1966,graham_bounds_1969}. The first PTAS was given by Hochbaum and Shmoys~\cite{hochbaum_using_1987} and had a running time of~$(n/\e)^{O((1/\e)^2)} = n^{O((1/\e)^2\log(1/\e))}$. This was improved to $n^{O((1/\e)\log^2(1/\e))}$ by Leung~\cite{leung_bin_1989}. Subsequent articles improve further the running time. In particular Hochbaum and Shmoys (see~\cite{hochbaum_approximation_1997}) and Alon et al.~\cite{alon_approximation_1997,alon_approximation_1998} obtain an \emph{efficient PTAS}\footnote{That is, a PTAS whose running time is $f(1/\e)\text{poly}(|I|)$ where $|I|$ is the encoding size of the input and $f$ is some function.} (EPTAS) with running time $2^{(1/\e)^{\text{poly}(1/\e)}} + O(n\log n)$. Alon et al.~\cite{alon_approximation_1997,alon_approximation_1998} consider general techniques that work for several objective functions, including all $L_p$-norm of the loads and maximizing the minimum machine load. 

The fastest PTAS known up to date for $P||C_{\max}$ achieves a running time of $2^{O((1/\e)^2\log^3(1/\e))} + O(n\log n)$~\cite{jansen_eptas_2010}. \textcolor{red}{More generally, this work gives an EPTAS for the case of related (uniform) machines, where each machine $i\in \mathcal{M}$ has a speed $s_i$ and assigning to $i$ job $j$ implies a processing time of $p_j/s_i$. For this more general case the running time is $2^{O((1/\e)^2\log^3(1/\e))} + \text{poly}(n)$}. For the simpler case of $P||C_{\max}$, the ILP can be solved directly since the number of variables is a constant. This can be done with Lentras' algorithm~\cite{lenstra_integer_1983}, or even with Kannan's algorithm~\cite{kannan_minkowskis_1987} that gives an improved running time. This technique yields a running time that is doubly exponential in $1/\e$. This was, in essence, the approach by Alon et al.~\cite{alon_approximation_1997,alon_approximation_1998} and Hochbaum and Shmoys~\cite{hochbaum_approximation_1997}. To lower the dependency on $1/\e$, Jansen~\cite{jansen_eptas_2010} uses a result by Eisenbrand and Shmonin~\cite{eisenbrand_caratheodory_2006} that implies the existence of a solution $x$ with support of size \textcolor{red}{at most $O(d\log(dT))=O((1/\e)\log^2(1/\e))$}. First guessing the support and then solving the ILP with $O((1/\e)\log^2(1/\e))$ integer variables and using Kannan's algorithm yields the desired running time of $2^{O((1/\e)^2\log^3(1/\e))} + O(n \log n)$. 

The configuration ILP has recently been studied in the context of the (1-dimensional) cutting stock problem. In this case, the dimension $d$ is constant, $T=1$, and $\pi$ is a rational vector. Moreover,  \textcolor{red}{$\pi$} and $b$ are part of the input. \textcolor{red}{Goemans and Rothvo\ss~\cite{goemans_polynomiality_2014} obtain an optimal solution in time $\log(\Delta)^{2^{O(d)}}$, where $\Delta$ is the largest number appearing in the denominator of $\pi_k$ or the multiplicities $b_k$}. This is achieved by first showing that there exists a pre-computable set $\tilde{Q} \subseteq Q$ with polynomial many elements, such that there exists a solution $x$ that gives all but constant (depending only on $d$) amount of weight to $\tilde{Q}$. We remark that applying this result to a rounded instance of $P||C_{\max}$ yields a running time that is doubly exponential on~$1/\e$. 

\subsection{Our Contributions}

Our main contribution is a new insight on the structure of the solutions of [conf-IP]. These properties are specially tailored to problems in which $T$ is bounded by a constant, which in the case of $P||C_{\max}$ can be guaranteed by rounding and scaling. The same holds for $Q||C_{\max}$ with a more complex rounding and case analysis. 

We first classify configurations by their support. We say that a configuration is \emph{simple} if its support is of size at most $\log(T+1)$, otherwise it is \emph{complex}. Our main structural result\footnote{We remark the resemblance of this structure to the result by Goemans and Rothvo\ss~\cite{goemans_polynomiality_2014}. Indeed, similar to their result, we can precompute a subset of configurations such that all but a constant amount of weight of the solution is given to such set. In their case the set is of cardinality polynomial on the input and is constructed by covering the integral solutions of the knapsack polytope by parallelepipeds. In our case, all but $O(d\log dT)$ weight is given to simple configurations.}  states that there exists a solution $x$ in which all but $O(d\log (dT))$ weight is given to simple configurations, the support is bounded by $O(d\log(dT))$ (as implied by Eisenbrand and Shmonin~\cite{eisenbrand_caratheodory_2006}) and no complex configuration has weight larger than~1.

\begin{restatable}[Thin solutions]{theorem}{thin}
% \begin{theorem}[Thin solutions]
\label{thm:thin}
 Assume that [conf-IP] is feasible. Then there exists a feasible solution $x$ to [conf-IP] such that:
 \begin{enumerate}
  \item if $x_c>1$ then the configuration $c$ is simple,
  \item the support of $x$ satisfies $|\supp{x}|\le 4(d+1)\log(4(d+1)T)$, and
  \item $\sum_{c \in Q_c} x_c\le 2(d+1)\log(4(d+1)T)$, where $Q_c$ denotes the set of complex configurations.
 \end{enumerate}
% \end{theorem}
\end{restatable}

We call a solution satisfying the properties of the theorem \emph{thin}. The theorem can be shown by iteratively applying a sparsification lemma that shows that if a solution gives a weight of two or more to a complex configuration, then we can replace this partial solution by two configurations with smaller support. The sparsification lemma is shown by a simple application of the pigeonhole principle. The theorem can be shown by mixing this technique with the theorem of Eisenbrand and Shmonin~\cite{eisenbrand_caratheodory_2006} and a potential function argument. 

As an application to our main structural theorem, we derive a PTAS for $P||C_{\max}$ by first guessing the jobs assigned to complex configurations. An optimal solution for this subinstance can be derived by a dynamic program. For the remaining instance we know the existence of a solution using only simple configurations. Then we can guess the support of such solution and solve the corresponding [conf-IP] restricted to the guessed variables. The main use of having simple configurations is that we can guess the support of the solution much faster, as the number of simple configuration is (asymptotically) smaller than the total number of configurations. The complete procedure takes time $2^{O((1/\e)\log^4(1/\e))} + O(n\log n)$. Moreover, using the rounding and case analysis of Jansen~\cite{jansen_eptas_2010}, we derive an mixed integer linear program that can be suitably decomposed in order to apply our structural result iteratively. This yields a PTAS \textcolor{red}{with a running time of $2^{O((1/\e)\log^4(1/\e))} + \text{poly}(n)$} for $Q||C_{\max}$. 

Similarly, we can extend our results to derive PTAS's for a larger family of objective functions as considered by Alon et al.~\cite{alon_approximation_1997,alon_approximation_1998}. Let $\ell_i$ denote the load of machine $i$, that is, the total processing time of jobs assigned to machine $i$ for a given solution. Our techniques then gives a PTAS with the same running time for the problem of minimizing the $L_p$-norms of the loads (for fixed $p$), and maximizing $\min_{i\in M} \ell_i$, among others. To solve this problem, we can round the instance and state an IP analogous to [conf-IP] but considering an objective function. However, the objective function prevents us to use the main theorem as it is stated. To get over this issue, we study several ILPs. In each ILP we consider $x_c$ to be a variable only if $c$ has a given load, and fix the rest to be some optimal solution. Applying to each such ILP Theorem~\ref{thm:thin}, plus some extra ideas, yields an analogous structural theorem. Afterwards, an algorithm similar to the one for makespan minimization yields the desired PTAS.%\jcom{To check this paragraph after we have written the section for the generalized objective functions. Probably makes sense to add more details as to what are exactly the objective functions that we can handle.}

From an structural point of view, our sparsification lemma has other consequences on the structure of the knapsack polytope and the LP-relaxation of the [conf-IP]. More precisely, we can show that any vertex of the convex hull of $Q$ must be simple. This, for example, helps us to upper bound the number of vertices by $2^{O(\log^2(T) + \log^2(d))}$. Moreover, we can show that the configuration-LP, obtained by replacing the integrality restriction in [conf-IP] by $x\ge0$, if it is feasible then admits a solution whose support consist purely of simple configurations. 
Due to space limitations we leave many details and proofs to the appendix.

\section{Preliminaries}
\label{sec:preliminaries}

We will use the following notation throughout the paper. By default $\log(\cdot)=\log_2(\cdot)$, unless stated otherwise. Given two sets $A,I$, we will denote by $A^I$ the set of all vectors indexed by~$I$ with entries in $A$, that is, $A^I= \{(a_i)_{i\in I}\,:\, a_i\in A \text{ for all } i\in I\}$. Moreover, for $A\subseteq \mathbb{R}$, we denote the support of a vector $a\in A^I$ as $\text{supp}(a)=\{i\in I: a_i\neq 0\}$.

We consider an arbitrary knapsack polytope $\mathcal{P} = \{c\in \mathbb{R}^{d}_{\geq 0}: \pi\cdot c \le T\}$ where $\pi \in \mathbb{Z}^d_{>0}$ is a non-negative integral (row) vector and $T$ is a positive integer. We assume without loss of generality that each coordinate \(\pi_k\) of $\pi$ is upper bounded by $T$ (otherwise $c_k=0$ for all $c\in \mathbb{Z}^d \cap \mathcal{P}$). We focus on the set of integral vectors in $\mathcal{P}$ which we denote by $Q = \mathbb{Z}^d \cap \mathcal{P}$. We call an element $c\in Q$ a \emph{configuration}. Given $b\in \mathbb{R}^d$, consider the problem of decomposing $b$ as a conic integral combination of $m$ configurations. That is, our aim is to find a feasible solution to [conf-IP], defined above.

% \begin{align}
%  \label{eq:ConicComb}
%  \text{[conf-IP]}\quad  & \sum_{c\in Q} c\cdot x_c =b,\\
%  \label{eq:NumberConf}
%  & \sum_{c\in Q} x_c = m,\\
%  &x_c \in \mathbb{N}_0 & \text{ for all }c\in Q. 
% \end{align}

% In this section we give our main technical contribution. In essence, we show that there exists an optimal solution $x \in \mathbb{N}_{0}^Q$ for [conf-IP] such that: (i) the support of $x$ contains only $\OO(d\log dT)$ configurations and (ii) if $x_c>1$ then $|\supp{c}|= \OO(\log T)$. 

% We start by modeling the problem with a \emph{configuration} IP formulation. Let $T$ be a guess on our optimal makespan such that $\LB\le T\le 2\LB{}$. We denote by $\mathcal{C}(T):= \{C\in \mathbb{N}^{P}: \size{C} \le T\}$. We will drop the dependency on $T$ when it is clear from the context. In the following IP the variable $x_C$ denotes the number of machines that follow configuration $C\in \mathcal{C}(T)$.

% \begin{align*}
%  \text{[Conf-IP]}\qquad &\sum_{C\in \mathcal{C}} x_C = m,\\
%  &\sum_{C\in \mathcal{C}} a(p,C)\cdot x_C = n_p & \text{for all } p\in P,\\ 
%  & x_C \in \mathbb{N} & \text{for all } C\in \mathcal{C}.
% \end{align*}

% For a given vector $x$, we call the support the subset $\text{supp}(x)=\{C\in \mathcal{C}: x_C\neq 0\}$. 

A crucial property of the [conf-IP] is that there is always a solution with a support of small cardinality. This follows from a Caratheodory-type bound obtained by Eisenbrand and Shmonin~\cite{eisenbrand_caratheodory_2006}. Since we will need the argument later, we state the result applied to our case and revise its (very elegant) proof. We split the proof in two lemmas. 

 For a given subset $A\subseteq Q$, let us denote by $x^A$ the indicator vector of $A$, that is $x^A_{c}=1$ if $c\in A$, and 0 otherwise. Let us also denote by $M$ the $(d+1)\times |Q|$ matrix that defines the system of equalities~\eqref{eq:ConicComb} and~\eqref{eq:NumberConf}.

\begin{lemma}[Eisenbrand and Shmonin~\cite{eisenbrand_caratheodory_2006}]
 \label{lm:2subsets}
 Let $x \in \mathbb{Z}_{\ge0}^{Q}$ be a vector such that $|\supp{x}| > 2(d+1)\log(4(d+1)T)$. Then there exist two disjoint sets $A,B$ with $\emptyset\neq A,B\subseteq \supp{x}$ such that $Mx^A=Mx^B$.
\end{lemma}
\begin{proof}
 %\jcom{This proof can go to the appendix for the conf version}
 Let $s:=|\supp{x}|$. Each coordinate of $M$ is smaller than $T$. Hence, for any $A\subseteq \supp{x}$, each coordinate of $M x^A$ is no larger than $|A|\cdot T\le sT$. Thus, $Mx^A$ belongs to $\{0,\ldots,sT\}^{d+1}$, and hence there are at most $(sT+1)^{d+1}=2^{(d+1)\log(sT+1)}$ different possibilities for vector $Mx^A$, over all possible subsets $A\subseteq \supp{x}$. On the other hand, there are $2^{s}$ many different subsets of $\supp{x}$. 
 
 We claim that $s> (d+1)\log(sT+1)$. Indeed, 
 since $s > 2(d+1)\log (4(d+1)T)$ then $T<2^{\frac{s}{2(d+1)}}/(4(d+1))$. Hence,
 
 \begin{align*}
   (d+1)\log(sT+1) &< (d+1) \log\left(\frac{s2^{\frac{s}{2(d+1)}}}{4(d+1)} + 1 \right)\\
   & \le (d+1)\log\left( 2^{\frac{s}{2(d+1)}}\left(\frac{s}{4(d+1)}+1\right)\right)\\
   &= (d+1)\left(\frac{s}{2(d+1)} + \log\left(\frac{s}{4(d+1)}+1\right)\right)\\
   & \le \frac{s}{2} + \frac{s}{4\ln(2)} < s,
 \end{align*}
 where the penultimate inequality follows since $\log(x)\le (x-1)/\ln(2)$ for all $x\geq 1$.

 We obtain that $2^{s}> 2^{(d+1)\log(sT+1)}$. Hence, by the pigeonhole principle there are two distinct subsets $A',B' \subseteq \supp{x}$ such that $Mx^{A'} = Mx^{B'}$. We can now define $A = A'\setminus B'$ and $B = B'\setminus A'$ and obtain $Mx^A = Mx^B$. It remains to show that $A,B\neq \emptyset$. Notice that if $A=\emptyset$ then $A'\subseteq B'$, and the last equality of $Mx^{A'} = Mx^{B'}$ implies that $|A'|=|B'|$. This is a contradiction since then $A'=B'$. We conclude that $A\neq \emptyset$. The proof that $B\neq \emptyset$ is analogous.
\end{proof}

\begin{lemma}[Eisenbrand and Shmonin~\cite{eisenbrand_caratheodory_2006}]
\label{lm:ES}
If [conf-IP] is feasible, then there exists a feasible solution $x$ such that $|\supp{x}| \le 2(d+1)\log(4(d+1)T)$.
\end{lemma}
\begin{proof}
 Let $x$ be a solution to [conf-IP] that minimizes $|\supp{x}| = s$. Assume by contradiction that $s > 2(d+1)\log(4(d+1)T)$. We show that we can find another solution $x'$ to [conf-IP] with $|\supp{x'}|<|\supp{x}|$, contradicting the minimality of $|\supp{x}|$. By Lemma~\ref{lm:2subsets}, there exist two disjoint subsets $A,B\in \supp{x}$ such that $Mx^A=Mx^B$. Moreover, let $\lambda = \min\{x_c: c\in A\}$. Vector $x':=x- \lambda x^A + \lambda x^B$ is also a solution to [conf-IP] and has a strictly smaller support since a configuration $c^* \in \arg\min\{x_c: c\in A\}$ satisfies $x'_{c^*}=0$.  
\end{proof}
\section{Structural Results}
 \label{sec:structural}
 
Recall that we call a configuration $c$ simple if $|\supp{c}|\le \log(T+1)$ and complex otherwise. An important observation to show Theorem~\ref{thm:thin} is that if $c$ is a complex configuration, then $2c$ can be written as the sum of two configurations of smaller support. This is shown by the following Sparsification Lemma.

\begin{lemma}[Sparsification Lemma] \label{lem-split}
 Let $c\in Q$ be a complex configuration. Then there exist two configurations $c_1,c_2\in Q$ such that
 \begin{enumerate}
  \item $\pi\cdot c_1 = \pi\cdot c_2 = \pi\cdot c$,
  \item $ 2c = c_1 + c_2$,
  \item $\supp{c_1}\subsetneq \supp{c}$ and $\supp{c_2}\subsetneq \supp{c}$.
  % $\supp{c_2}\cap\supp{c_1}=\emptyset$. Unfortunately we can not get this property in the case $c^{S}_i = c^{R}_i0 0$ . Does that make any problems later?
 \end{enumerate}
 \end{lemma}
\begin{proof}
 %The number of subsets of $\supp{c}$ is $2^{|\supp{c}|} > 2^{\log T}= T$. For each subset $S\subseteq \supp{c}$, we consider a configuration $c^S\in Q$ such that $c^S_i=c_i$ if $i\in S$ and $c^S=0$ otherwise. Consider now the collection of vectors $V:=\{c^S: S \subseteq \supp{c}\}$, and notice that $|V|>T$.
 
Consider for each subset $S\subseteq \supp{c}$, a configuration $c^S\in Q$ such that $c^S_i=c_i$ if $i\in S$ and $c^S=0$ otherwise. As the number of subsets of $\supp{c}$ is $2^{|\supp{c}|}$, and $c^R\neq c^S$ if and only if $R\neq S$, the collection of vectors $V:=\{c^S: S \subseteq \supp{c}\}$ has cardinality $|V| = 2^{|\supp{c}|}$.

On the other hand, for any  vector $c^S \in  V$ it holds that $\pi\cdot c^S\le \pi\cdot c \le T$. Hence, $\pi\cdot c^S\in\{0,1\ldots,T\}$ can take only $T+1$ different values. Using that $c$ is a complex configuration and hence $2^{|\supp{c}|} > 2^{\log (T+1)} = T+1$, the pigeonhole principle ensures that there are two different non-empty configurations $c^S,c^R \subseteq V$ with $\pi\cdot c^S = \pi\cdot c^{R}$. By removing the intersection, we can assume w.l.o.g. that $S$ and $R$ have no intersection. We define $c_1 = c - c^S + c^R$ and $c_2 = c - c^R + c^S$, which satisfy the properties of the lemma as
\begin{align*}
& \pi\cdot c_1 = \pi\cdot c - \pi\cdot c^S + \pi\cdot c^R = \pi\cdot c \quad \text{ and}\\
& 2c = c - c^S + c^R + c - c^R + c^S = c_1 + c_2.  %\qedhere 
\end{align*}
Since $\supp{c_1}\subseteq \supp{c}\setminus S$ and $\supp{c_2} \subseteq \supp{c}\setminus R$, property 3 is satisfied.
\end{proof}

With Lemma \ref{lem-split} we are ready to show Theorem~\ref{thm:thin}. For the proof it is tempting to apply the lemma iteratively, replacing any complex configuration that is used twice by two configurations with smaller support. This can be repeated until there is no complex configuration taken multiple times. Then we can apply the technique of Lemma~\ref{lm:ES} to the obtained solution to bound the cardinality of the support. However, the last step might break the structure obtained if the solution implied by Lemma~\ref{lm:ES} uses a complex configuration more than once. In order to avoid this issue we consider a potential function. We show that a vector minimizing the chosen potential uses each complex configuration at most once, and that the number of complex configurations in the support is bounded. Finally, we apply the techniques from Lemma~\ref{lm:ES} restricted to variables corresponding to simple configurations. %In what follows we restate the theorem and provide its proof.

%\thin*

\begin{proof}[Proof of Theorem~\ref{thm:thin}]
Consider the following potential function of a solution $x \in \mathbb{Z}_{\ge 0}^Q$ of [conf-IP],
\begin{align*}
 \Phi(x) = \sum_{\text{complex config. } c} x_c |\supp{c}| .
\end{align*}
Let $x$ be a solution of [conf-IP] with minimum potential $\Phi(x)$, which is well defined since the set of feasible solutions has finite cardinality. We show two properties of $x$.

 \noindent {\bf P1:} $x_c\le1$ for each complex configuration $c\in Q$.
 
Assume otherwise. Consider the two configurations $c_1$ and $c_2$ implied by the previous lemma. We define a new solution $x'_{e} = x_{e}$ for $e\not\in \{c,c_1,c_2\}$, $x'_{c_1}= x_{c_1}+1$, $x'_{c_2} = x_{c_2}+1$ and $x'_c = x_c-2$. Since $|\supp{c_1}|<| \supp{c}|$ and $|\supp{c_2}|< |\supp{c}|$, we obtain that $\Phi(x') < \Phi(x)$ which contradicts the minimality of $\Phi(x)$.

\noindent {\bf P2:} The number of complex configurations in $\supp{x}$ is at most $2(d+1)\log(4(d+1)T)$.

Let $\tilde{x}$ be the vector defined as $\tilde{x}_c=x_c$ if $c\in Q$ is complex, and $\tilde{x}=0$ if $c\in Q$ is simple. Then Lemma~\ref{lm:2subsets} implies that there are exist two disjoint subsets $A,B\subseteq\supp{\tilde{x}}$ of complex configurations such that $Mx^A = Mx^B$. Thus, the solution $x' = x - x^A + x^B$ and the solution $x'' = x - x^B + x^A$ are feasible for [config-IP]. By linearity, the potential function on the new solutions are $\Phi(x') = \Phi(x) - \Phi(x^A) +  \Phi(x^B)$ or respectively  $\Phi(x'') = \Phi(x) -  \Phi(x^B) + \Phi(x^A)$. If $\Phi(x^A) > \Phi(x^B)$ or $\Phi(x^B) > \Phi(x^A)$ then we have constructed a new solution with smaller potential, contradicting our assumption on the minimality of $\Phi(x)$. We conclude that $\Phi(x^B) = \Phi(x^A)$ and thus $\Phi(x)=\Phi(x')$. By construction of $x'$, we obtain that $x_c' > x_c\ge 1$ for any complex configuration $c\in B$. Having multiplicity $\geq 2$ for a complex configuration $c$, we can proceed as in Case 1 to find a new solution with decreased potential, which yields a contradiction.
 
 Given these two properties, to conclude the theorem it suffices to upper bound the number of simple configurations by $2(d+1)\log(4(d+1)T)$. Suppose this property is violated, then we find two sets $A,B\subseteq \supp{x}$ of simple configurations (see Lemma \ref{lm:2subsets}) with $Mx^A = Mx^B$ and proceed as in Lemma~\ref{lm:ES}. Since Lemma~\ref{lm:ES} is only applied to simple configurations, properties P1 and P2 continue to hold and the theorem follows.
\end{proof}

 Our techniques, in particular our Sparsification Lemma, imply two corollaries on the structure of the knapsack polytope and the LP-relaxation implied by the [conf-IP].
 
 \begin{corollary} \label{cor:vertices}
  Every vertex of $\text{conv.hull}(Q)$ is a simple configuration. Moreover, the total number of simple configurations in $Q$ is upper bounded by $2^{O(\log^2(T) + \log^2(d))}$ and thus the same expression upper bounds the number of vertices of $\text{conv.hull}(Q)$.
 \end{corollary}
\begin{proof}
Consider a complex configuration $c \in Q$. By Lemma \ref{lem-split} we know that there exist $c_1, c_2 \in Q$ with $c_1, c_2 \neq c$ such that $2c = c_1 + c_2$. Hence, $c$ is not a vertex of $Q$ as it can be written as a convex combination $c = c_1/2 + c_2/2$.
 
  %By taking $m=1$, [conf-LP] is feasible if and only if $b\in \text{conv.hull}(Q)$. The fact that every vertex of $\text{conv.hull}(Q)$ is a simple configuration follows then by the previous theorem since $b$ can be written as a convex combination of simple configurations. 
To bound the number of simple configurations fix a set $D\subseteq \{1,\ldots,d\}$. Notice that the number of configurations $c$ with $\supp{c}=D$ is at most $T^{|D|}$. For simple configurations it suffices to take $D$ with cardinality  at most $\log(T+1)$. Since the number of subsets $D\subseteq \{1,\ldots,d\}$ with cardinality $i$ is ${d \choose i}$, we obtain that the number of simple configurations is at most
\begin{align*}
&\sum_{i=0}^{\lfloor \log(T+1) \rfloor}{d \choose i} \times(T+1)^{ \log(T+1)}
\leq (\log (T+1)+1) d^{\log (T+1)} \times (T+1)^{\log(T+1)} \\
&= 2^{\log(\log(T+1)+1) + \log (d) \log (T+1)} \times 2^{\log (T+1)\log (T+
1)} = 2^{O(\log^2(d)+ \log^2(T))}. \qedhere
\end{align*}
% The term $\sum_{i=0}^{\lceil \log(T+2) \rceil}{d \choose i}$ is a bound on the number of possible supports of a simple configuration, and $(T+1)^{\lceil \log(T+2) \rceil}$ bounds the number of configurations of given support.
 %${d+\log (T+2) \choose \log (T+2)}\times(T+1)^{\log(T+2)}\le ((d+\log (T+2))(T+1))^{\log (T+2)}$, where ${d+\log T \choose \log T}$ is a bound on the number of possibles supports of a simple configuration, and $(T+1)^{\log(T)}$ bounds the number of configurations with a given support.
\end{proof}
 
The following corollary follows as each complex configuration can be represented by a convex combination of simple configurations.
\begin{corollary} \label{cor:vertices2}
  Let [conf-LP] be the LP relaxation of [conf-IP], obtained by changing each constraint $x_c\in \mathbb{Z}_{\ge 0}$ to $x_c\ge 0$ for all $c\in Q$. If the LP is feasible then there exists a solution $x$ such that each configuration $c\in \supp{x}$ is simple.
 \end{corollary}
\begin{proof}
Consider a solution $x$ of [conf-LP]. Assume that there exists $c \in Q$ such that $c$ is complex and $x_c > 0$. Then by the previous corollary, configuration $c$ can be written as $c = \sum_{q\in Q}\lambda_q q$, where  $\sum_{q\in Q} \lambda_q=1$,  $\lambda_q\ge 0$ for all $q\in Q$, and $\lambda_q=0$ if $q\in Q$ is complex. Consider a new solution $x'$ defined as 
\[
 x'_{q} = \begin{cases}
          0 &\text{if } q = c,\\
         x_{q} + \lambda_q\cdot x_c & \text{if } q\neq c.
        \end{cases}
\]

This new solution is also feasible for [conf-LP]. As $x'_c =0$, the number of complex configurations in the support of the solution is reduced by $1$.
This procedure can be repeated until we have a solution $\hat{x}$ whose support contains only simple configurations.
\end{proof}

\section{Applications to Scheduling on Parallel Machines}
\label{sec:PCmax}

In what follows we show how to exploit the structural insights of the previous section to derive faster algorithms for parallel machines scheduling problems. We start by considering $P||C_{\max}$, where we seek to assign a set of jobs $\mathcal{J}$ with processing times $p_j\in \mathbb{Z}_{>0}$ to a set $\mathcal{M}$ of $m$ machines. For a given assignment $a:\mathcal{J} \mapsto \mathcal{M}$, we define the load of a machine $i$ as $\sum_{j: a(j)=i} p_j$ and the \emph{makespan} as the maximum load of jobs over all machines, which  is the minimum time needed to complete the execution of all jobs on the processors. The goal is to find an assignment $\mathcal{J} \mapsto \mathcal{M}$ that minimizes the makespan.

We first follow well known rounding techniques~\cite{alon_approximation_1997,alon_approximation_1998,hochbaum_using_1987,hochbaum_approximation_1997}. Consider an error tolerance $0<\e<1/3$ such that $1/\e^2$ is an integer. To get an estimation of the optimal makespan, we follow the standard dual approximation approach. First, we can use, e.g., the 2-approximation algorithm by Graham~\cite{graham_bounds_1966} to get an initial guess of the optimal makespan. Using binary search, we can then  estimate the optimal makespan within a factor of $(1+ \e)$ in $O(\log(1 / \e))$ iterations. Therefore, it remains to give an algorithm that decides for a given makespan $T$, if there exists an assignment with makespan $(1+ O(\e))T$ or reports that there exists no assignment with makespan $ \leq T$.

For a given makespan $T$ we define the set of big jobs $\Jbig = \{ j \in \mathcal{J} : p_j \geq \e T\}$ and the set of small jobs $\Jsmall = \mathcal{J} \setminus \Jbig$. The following lemma shows that small jobs can be replaced from the instance by adding big jobs, each of size $\e T$, as placeholders. Let $S$ be the sum of processing times of jobs in $\Jsmall$ and let $S^*$ denote the next value of $S$ rounded up to the next multiple of $\e T$, that is, $S^* = \e T \cdot \lceil S/(\e T)\rceil$. We define a new instance containing only big jobs by $\Jround = \Jbig \cup \Jnew$, where $\Jnew$ contains $S^* / (\e T)\in\mathbb{N}$ jobs of size~$\e T$.

\begin{lemma}\label{lem:rounding1}
Given a feasible assignment $a: \mathcal{J} \mapsto \mathcal{M}$ of jobs with makespan $T$. Then there exists a feasible assignment $a_B: \Jround \mapsto \mathcal{M}$ of makespan $T^* \leq (1+\e)T$. Similarly, an assignment of jobs in $\mathcal{J}^*$ of makespan $T^*$ can be transformed to an assignment of $\mathcal{J}$ of makespan at most $(1+\e)T^*$.
\end{lemma}
\begin{proof}
We modify the assignment $a$ of jobs in $\mathcal{J}$ by replacing the set of small jobs on each machine by jobs in~$\Jnew$. Let~$S_i$ be the total processing time of small jobs assigned to machine~$i$. Then the small jobs are replaced by (at most) $S_{i}^* / (\e T)$ jobs in $\Jnew$, where $S_{i}^*$ denotes the value of $S_i$ rounded up to the next multiple of $\e T$. As $\sum \frac{S_{i}^*}{\e T} \geq \lfloor \sum \frac{S_i}{\e T} \rfloor = \frac{S^*}{\e T}$, the new solution processes all jobs in $\Jnew$ and the load on each machine increases hence by at most $\e T$.
Having an assignment for the big jobs $\mathcal{J}^*$, we can easily obtain a schedule for jobs $\mathcal{J}$, by adding the small items greedily into the space of the placeholder jobs $\Jnew$.
\end{proof}

%added greedily by the LPT rule (see \cite{}) to an instance without increasing the makespan too much. The LPT rule places jobs succesively on the machine with the smallest load.
%\begin{lemma}[\cite{}]
%Given an assignment $a_B: \Jround \mapsto \mathcal{M}$ of big jobs with makespan $T'$. If there exists a feasible assignment $a: \mathcal{J} \mapsto \mathcal{M}$ with makespan $T'$, then adding small jobs $\mathcal{J}_s$ greedily (according o the LPT rule) gives a schedule with makespan $T^*$ with $T^* \leq (1+ \e) T'$.
%\end{lemma}

By scaling the processing times of jobs in $\Jround$, we can assume that the makespan $T$ has value~$1/\e^2$. Also notice that we can assume that $p_j\le T$ for all $j$, otherwise we cannot pack all jobs within makespan $T$. This implies that each job $j \in \Jround$ has a processing time of $1/\e \le p_j \le 1/\e^2$. In the following we give a transformation of big jobs in $\Jround$ by rounding their processing times. We first round the jobs to the next power of $1+\e$ as $p_j'= (1+\e)^{\lceil \log_{(1+\e)}p_j\rceil}$, and thus all rounded processing times belong to $\Pi'=\{(1+\e)^{k}\,:\, 1/\e \le (1+\e)^k \le (1+\e)/\e^2 \text{ and } k\in \mathbb{N}\}$. We further round processing times $p_j'$ to the next integer $\bar{p}_j=\lceil p_j'\rceil$ and define a new set $\Pi= \{\lceil p \rceil\,:\, p\in \Pi'\}$. Notice that $\Pi$ only contains integers and $|\Pi| \le |\Pi'| \in O((1/\e)\log (1/\e))$. 

\begin{lemma}
\label{lm:rounding}
If there is a feasible schedule of jobs $\Jround$ with processing times $p_j$ onto $m$ machines with makespan $T^*\le (1+\e)T$, then there is also a feasible schedule of jobs $\Jround$ with rounded processing $\bar{p}_j$ with a makespan of at most $(1+5\e)T$. Furthermore, the number of different processing times is at most $|\Pi| \in O((1/\e)\log (1/\e))$.
%Let $\mathcal{J}'$ be the set of jobs from $\Jround$ with rounded processing times $P'$.
%$\{\frac{1}{\e^2},\frac{1}{\e^2}+1,\ldots,\frac{2}{\e^2}\}.$ An optimal solution to $I'$ can be transformed in linear time to a $(1+3\e)$-approximate solution solution to instance $I$.
\end{lemma}
\begin{proof}
Consider a feasible schedule of jobs in $\Jround$ with processing times $p_j$ onto $m$ machines with makespan $T^*$. Let $J_{i_1}, \ldots , J_{i_r}$ be the set of jobs processed on machine $i$ i.e. $a(J_{i_k}) = i$ for $k = 1, \ldots , r$. Then $\sum_{j=1}^r p'_j \leq \sum_{j=1}^r (1+ \e) p_j \leq (1+\e) T^*$.
Hence, the same assignment $a$ with processing times $p'_j$ yields a makespan of at most $(1+\e)T^*\le (1+\e)^2T = 1/\e^2 + 2/\e +1$. Since $p'_j \geq p_j \geq 1/\e$, on every machine are at most $1/\e + 2$ jobs. Hence, rounding the processing times $p'_j$ of each job to the next integer increase the load on each machine by at most $1/\e+2$. Recalling that $\e <1/3$, we obtain a feasible schedule with makespan at most $(1+\e)T^*  + 1/\e+2 \le 1/\e^2 + 3/\e + 3 < T + 5\e T$.
\end{proof}

In what follows we give an algorithm that decides in polynomial time the existence of a solution for instance $\Jround$ with processing times $\bar{p}_j$ and makespan $\bar{T}=\lfloor (1+5\e)T\rfloor$. We call numbers in $\Pi$ by $\pi_1,\ldots , \pi_d$ and define the vector $\pi =(\pi_1,\pi_2,\ldots,\pi_d)\in \mathbb{N}^d$ of rounded processing times. We consider \emph{configurations} to be vectors in $Q= \mathcal{P} \cap \mathbb{Z}^d$, where $\mathcal{P} = \{ c \in \mathbb{R}^{d}_{\geq 0}: \pi \cdot c \le \bar{T} \}$ is a knapsack polytope (see Section~\ref{sec:structural}). As before, we say that a configuration is simple if $|\supp{c}|\le \log (\bar{T}+1)$, and complex otherwise. For a given assignment of jobs to machines, we say that a machine follows a configuration $c$ if $c_k$ is the number of jobs of size $\pi_k$ assigned to the machine. We denote by $Q_c \subseteq Q$ the set of complex configurations and by  $Q_s \subseteq Q$ the set of simple configurations.

Let $b_k$ be the number of jobs of size $\pi_k$ in the instance $\Jround$ (with processing times $\bar{p}$). Consider an ILP with integer variables $x_c$ for each $c\in Q$, which denote the number of machines that follow configuration $c$. With these parameters the problem of scheduling all jobs in a solution of makespan $\bar{T}$ is equivalent to finding a solution to [conf-IP]. To solve the ILP we use, among other techniques, Kannan's algorithm~\cite{kannan_minkowskis_1987} which is an improvement on the algorithm by Lenstra~\cite{lenstra_integer_1983}. The algorithm has a running time of $2^{O(N\log N)}s$ where $N$ is the number of variables and $s$ is number of bits used to encode the input of the ILP in binary. %For our case this yields a running time of $2^{2^{\text{poly}(1/\e)}}\log n$. %Following the approach in~\cite{jansen_eptas_2010} we can improve the running time by first guessing the support of an optimal solution, which takes time $2^{O((1/\e^2)\log(1/\e))}$. This can be done since the theorem by Eisenbrand and Shmoning~\cite{eisenbrand_caratheodory_2006} applied to [conf-IP] guarantees the existence of a solution with support of size $O((1/\e)\log^2(1/\e))$. Then solving the IP restricted to the guessed support in time $2^{O((1/\e)\log^3(1/\e))}\log(n)$ with an algorithm for IP in constant dimension~\cite{kannan_minkowskis_1987}. This yields a running time of $2^{O(
%(1/\e^2)\log^3(1/\e)))}$. In what follows, we show how to use the extra structure of thin solutions to give an improved running time of $2^{O(1/\e)\log^4(1/\e))}$.%, which is close to the lower bound of $2^{O(\frac{1}{\e}^{1-\e})}$ shown in \cite{}.

By Theorem~\ref{thm:thin}, if [conf-IP] is feasible then there exists a thin solution. In particular if one configuration $c$ is used by more than one machine then $c$ is simple, and the total number of used configurations is $4(d+1)\log(4(d+1)\bar{T})\in O((1/\e)\log^2 (1/\e))$. Additionally, the number of machines following a complex configurations is at most $2(d+1)\log(4(d+1)\bar{T})\in O((1/\e)\log^2 (1/\e))$. We consider the following strategy to decide the existence of a schedule of makespan $\bar{T}$.

\begin{algorithm}
\label{alg:PTASparallelCmax}
\hspace{4cm}\begin{enumerate}
\item For each processing time $\pi_k$, guess the number $b^c_k \leq b_k$ of jobs covered by complex configurations.% $Q_c \subseteq Q$.
\item Find a minimum number of machines $m^c$ to schedule jobs $b^c$ with makespan $\bar{T}$.%This can be done using dynamic programming \cite{}. 
\item Guess the support of simple configurations ${\bar Q}_s \subseteq Q_s$ used by a thin solution, with $|{\bar Q}_s| \leq 4(d+1)\log(4(d+1)\bar{T})\in O((1/\e)\log^2 (1/\e))$.
\item Solve the ILP restricted to configurations in $\bar{Q}_s$:
\begin{align*}
 & \sum_{c\in \bar{Q}_s} c\cdot x_c =b-b^c,\\
 & \sum_{c\in {\bar Q}_s} x_c = m-m^c,\\
 &x_c \in \mathbb{Z}_{\ge 0} & \text{ for all }c\in \bar{Q}_s. 
\end{align*}
\end{enumerate}
\end{algorithm}

One of the key observations to prove the running time of the algorithm is that the number of simple configurations $|Q_s |$ is bounded by a quasi polynomial term:
\[
	|Q_s| \le 2^{O(\log^2 (1/\e))}.
\]
This follows easily by Corollary \ref{cor:vertices}, using that $|\bar{T}| \in O(1/\e^2)$ and $d = |\Pi| \in O((1/\e)\log (1/\e))$.

\begin{lemma} \label{lm:runningtime}
Algorithm~\ref{alg:PTASparallelCmax} can be implemented with a running time of $2^{O((1/\e)\log^4(1/\e))} \log(n)$.
\end{lemma}
\begin{proof}
In step 1, the algorithm guesses which jobs are processed on machines following a complex configurations. Since each configuration contains at most $O(1/\e)$ jobs, there are at most $O(m^c/\e)=O((1/\e^2)\log^2 (1/\e))$ jobs assigned to such machines. For each size $\pi_k \in \Pi$, we guess the number $b_k^c$ of jobs of size $\pi_k$ assigned to such machines. Hence, we can enumerate all possibilities for jobs assigned to complex machines in time $2^{O((1/\e)\log^2 (1/\e))}$. After guessing the jobs, we can assign them to a minimum number of machines in step 2 (with makespan $\bar{T}$) with a simple dynamic program that stores vectors $(\ell,z_1,\ldots,z_d)$ with $z_k\le b_k^c$ being the number of jobs of size $\pi_k$ used in the first $\ell\le m^c$ processors~\cite{jansen_scheduling_2011}. \textcolor{red}{The size of the dynamic programming table is $O(m^c\prod_{k=1}^d (b_k^c+1))$. For any vector $(\ell,z_1,\ldots,z_d)$, determining whether it corresponds to a feasible solution can be done by checking all vectors of the type $(\ell-1,z_1',\ldots,z_d')$ for $z_k'\le z_k$. Thus, the running time of the dynamic program is $O(m^c [\prod_{k=1}^d (b_k^c+1)]^2)$.} Since $b_k^c \in O((1/\e^2)\log^2 (1/\e))$ for each $k$, recalling that $m^c\in O((1/\e)\log^2 (1/\e))$, and that $d=|\Pi| \in O((1/\e) \log(1/\e))$, we obtain that step 2 can be implemented with $2^{O((1/\e)\log^2 (1/\e))}$ running time.
% 
% \begin{observation}
%  It holds that $|Q_s|\le 2^{O(\log^2 (\frac{1}{\e}))}$.
% \end{observation}

%By Corollary \ref{cor:vertices}, we know that the number of simple configurations is bounded by $2^{O(\log^2(\bar{T}) + \log^2(d))}$. Using that $|\bar{T}| \in O(1/\e^2)$ we obtain that $|Q_s|= 2^{O(\log^2 (1/\e))}$.

%To obtain the bound claimed in the observation, it suffices to notice that the support of a configuration $c\in Q_s$ can take at most ${|P|+\log(T) \choose \log(T) }= 2^{O(\log^2 (\frac{1}{\e}))}$ many different possibilities. For a given support $S$ with $|S|\le \log(T)$, there are at most $(\frac{1}{\e}+1)^{|S|}\le 2^{O(\log^2 (\frac{1}{\e}))}$ many configurations $c$ with $\supp{c}=S$. Then, $|Q_s|\le 2^{O(\log^2 (\frac{1}{\e}))}$. We remark that this number is considerably smaller than the total number of configurations $|Q|= 2^{O(\frac{1}{\e}\log^2 (\frac{1}{\e}))}$. This is indeed the fact that will allow us to solve the IP in step 3 faster. 

In step 3, our algorithm guesses the support of a thin solution $x$. Recall that if $x$ is thin then $|\supp{x}|\le 4(d+1)\log(4(d+1)\bar{T}) = O((1/\e)\log^2 (1/\e))$. Let $D=4(d+1)\log(4(d+1)\bar{T})$. Then this guess can be done in time

\[
\sum_{i=0}^{D} {|Q_s| \choose i} \le (D+1)  |Q_s|^{D} \le 2^{O((1/\e)\log^4 (1/\e))}.
\]
We remark that for this step is that thin solutions are particularly useful. Indeed, guessing the support on the original ILP takes time $2^{O((1/\e)^2 \log^3 (1/\e))}$.

In step 4, the number of variables of the restricted ILP is $4(d+1)\log(4(d+1)\bar{T})=O((1/\e)\log^2 (1/\e))$. Moreover, the size of the input can be bounded by $O((1/\e^2)\log^3(1/\e)\log(n))$. Running Kannan's algorithm \cite{kannan_minkowskis_1987} to solve the ILP takes time $2^{O((1/\e)\log^3(1/\e))}\log(n)$. Hence, the total running time of our algorithm can be bounded by $2^{O((1/\e)\log^4(1/\e))}\log(n)$. 
\end{proof}

%Let $b_k''= b_k - b_k'$ be the number of jobs of size $p_k$ left in the residual instance, and $m''\le m$ the machines left. We know that these jobs must be assigned to simple configurations. Let $Q_s\subseteq Q$ be the set of all simple configurations in $Q$. Then we can solve the residual problem with the following IP:

%Since we need to do this for each possible value of $T\in\{\frac{1}{\e^2},\frac{1}{\e^2}+1,\ldots,\frac{2}{\e^2}\}$, plus the time for rounding as done in Lemma~\ref{lm:rounding}, we obtain that the total running time is $2^{O(\frac{1}{\e}\log^4(\frac{1}{\e}))}\log(n)+n$. 
Putting all pieces together, we conclude with the following theorem.

\begin{theorem} \label{thm:main_runningtime}
 The minimum makespan problem on parallel machines $P||C_{\max}$ admits an EPTAS with running time $2^{O((1/\e)\log^4(1/\e))}+O(n\log n)$.
\end{theorem}
\begin{proof}
Consider a scheduling instance with job set $\mathcal{J}$, processing times $p_j$ for $j\in \mathcal{J}$ and machine set $\mathcal{M}$. 
%First sort the jobs by their processing time in time $O(n\log n )$. 
The greedy algorithm by Graham to obtain a 2-approximation can be implemented in $O(n\log n)$. After guessing the makespan $T$, the processing times are sorted and rounded as described in Lemma \ref{lm:rounding}. The rounding step can easily be implemented in $O(n)$ time. Applying Algorithm \ref{alg:PTASparallelCmax} after the rounding needs, according to Theorem \ref{lm:runningtime}, a running time of $2^{O((1/\e)\log^4(1/\e))} \log(n)$ time. Since there are at most $O(\log(1/\e))$ many guessing rounds for the makespan, we obtain a total running time of $O(n\log n + \log(1/\e)\cdot n) + 2^{O((1/\e)\log^4(1/\e))} \log(n)$.

If $n \le 2^{\frac{1}{\e}\log^4(\frac{1}{\e})}$ then the running time is upper bounded by $2^{O(\frac{1}{\e}\log^4(\frac{1}{\e}))}$, otherwise, the running time is at most $O(n\log n)$. In any case, the running time can be bounded by $2^{O(\frac{1}{\e}\log^4(\frac{1}{\e}))}+O(n\log n)$.
\end{proof}

\subsection{Extension to other objectives}
\label{sec:genObjt}

We now consider a more general family of objective functions defined by Alon et al.~\cite{alon_approximation_1997,alon_approximation_1998}. For a fixed function $f:\mathbb{R}_{\ge0} \rightarrow \mathbb{R}_{\ge0}$, we consider the following two objective functions: 

\begin{center}
 (I) $\min \sum_{i\in \mathcal{M}} f(\ell_i)$ \hspace{1cm }(II) $\min \max_{i\in \mathcal{M}} f(\ell_i)$,
 \end{center}
where $\ell_i$ denotes the load of machine $i$. Analogously, we study maximization versions of the problems

\begin{center}
 (I') $\max \sum_{i\in \mathcal{M}} f(\ell_i)$ \hspace{1cm }(II') $\max \min_{i\in \mathcal{M}} f(\ell_i)$,
 \end{center}
 For the minimization versions of the problem we assume that $f$ is convex, while for (I') and (II') we assume it is concave. Moreover, we will need that the function satisfies the following sensitivity condition.
\begin{condition}\label{cond:sensitivity}
 For all $\e >0 $ there exists $\delta=\delta(\e)>0$ such that for all $x,y\in \mathbb{R}_{\ge 0}$,
 \[
     (1-\delta)y \le x \le (1+\delta)y \quad \Rightarrow \quad (1-\e)f(y) \le f(x) \le (1+\e)f(y).
 \]
\end{condition}

Alon et al. showed that each problem in that family admits a PTAS with running time $h(\e) + O(n\log n)$, where $h(\e)$ is a constant term that depends only on $\e$. Moreover, if $\delta(\e)$ in the condition further satisfies that $1/(\delta(\e))\in O(1/\e)$, the running time is $2^{(1/\e)^{\text{poly}(1/\e)}} + O(n\log n)$. In what follows we show how to improve this dependency. Since $1/(\delta(\e))\in O(1/\e)$, we know that, for small enough $\e$, there exists a constant $\gamma$ (independent of $\e$ and $\delta$) such that $1/\delta\le \gamma/\e$. Moreover, we can assume w.l.o.g. that $\delta\le \e$, and thus $\delta\le \e \le \gamma \delta$. 

It is worth noticing that many interesting functions belong to this family. In particular (II) with $f(x)=x$ corresponds to the minimum makespan problem, (I) with $f(x) = x^p$, for constant $p$, corresponds to a problem that is equivalent to minimizing the $L_p$-norm of the vector of loads. Similarly, (II') with $f(x)=x$ corresponds to maximizing the minimum machine load. Notice that for all those objectives we have that $1/\delta = O(1/\e)$.

The techniques of Alon et al.~\cite{alon_approximation_1998} are based on a rounding method and then solving an ILP. We based our results in the same rounding techniques. Consider an arbitrary instance of a scheduling problem on identical machines with objective function (I), (II), (I') or (II'). Their first observation is that, if $L=\sum_j p_j/m$ is the average machine load, then a job with $p_j\ge L$ is scheduled alone on a machine in an optimal \textcolor{red}{solution~\cite{alon_approximation_1998}}. Hence, we can remove such job and a machine from the instance. In what follows, we assume without loss of generality, that $p_j< L$ for all $j$. For the sake of brevity, we summarize the rounding techniques of Alon et al. in the following theorem.

\begin{theorem}[Alon et al.~\cite{alon_approximation_1998}]
\label{thm:roundAlon}
 Consider an instance for the scheduling problem with job set~$\mathcal{J}$, identical machines~$\mathcal{M}$, and processing times $p_j$ for $j\in \mathcal{J}$ such that $p_j< L$ for all $j$. There exists a linear time algorithm that creates a new instance $I'$ with job set $\mathcal{J'}$, machine set $\mathcal{M}$, and processing times $p_j'$. \textcolor{red}{Moreover, there is an integer $\lambda\ge 1/\delta$ with $\lambda \in O(1/\delta)$ such that the new instance satisfies the following:}
 \begin{enumerate}
  \item Each job $j$ in $I'$ has processing time  $L/\lambda \le p_j' \le L$, and $p_j'$ is a integer multiple of $L/\lambda^2$.
  \item \textcolor{red}{If $L'=\sum_{j} p'_j/m$ then $L\le L'\le (1+2/\lambda)L$.}
%   \item There exists an optimal solution of $I'$ such that for each machine $i$ its load $\ell_i$ satisfies $L'/2 \le \ell_i \le 2L'$, where $L'=\sum_{j}p_j'/m$.
  \item Let $\OPT$ and $\OPT'$ be the optimal value of instances $I$ and $I'$, respectively. Then $(1-\e)\OPT\le \OPT' \le (1+\e)\OPT$.
  \item There exists a linear time algorithm that transforms a feasible solution for instance $I'$ with objective value $V$ \kkcom{$V$ and $V'$ correct here? Seems they need to be changed} \jcom{Corrected.} to a feasible solution for $I$ with objective value $V'$ such that \textcolor{red}{$(1-\e)V\le V' \le (1+\e)V$.} 
 \end{enumerate}
\end{theorem}

Given this result, it suffices to find a $(1+\e)$-approximate solution for instance $I'$. \kkcom{Maybe we should briefly mention what happens to the small jobs} \jcom{Agreed that it would be better but this is not critical. Let's do it for the journal version.} To do so, we further round the processing times as in the previous section by defining $\bar{p}_j$ as the value $(1+\delta)^{\lceil \log_{1+\delta} p_j'\rceil}$ rounded up to the next multiple of $L/\lambda^2$ for all $j\in \mathcal{J}'$. Notice that \textcolor{red}{$\bar{p}_j\le (1+\delta)^{\lceil \log_{1+\delta} p_j'\rceil} + L/\lambda^2 \le (1+\delta)p_j' + L/\lambda^2 \le (1+\delta)p_j' + p_j'/\lambda \le (1+2\delta)p_j' \le (1+\delta)^2p_j'$}. \textcolor{red}{Hence, for any assignment that gives a load $\ell_i$ on machine $i$ for $p'_j$, the same assignment has a load $\bar{\ell}_i$ with $\ell_i\le \bar{\ell}_i \le (1+\delta)^2\ell_i$. By Condition~\ref{cond:sensitivity} we conclude that the} new optimal value $\overline{\OPT}$ satisfies that $(1-O(\e))\OPT\le\overline{\OPT}\le(1+O(\e))\OPT$.

Let $\Pi=\{\pi_1,\ldots,\pi_d\}$ be the distinct values that the processing times $\bar{p}_j$ can take. Notice that $d=|\Pi| = O((1/\delta)\log(1/\delta))$. We consider the knapsack polytope with capacity \textcolor{red}{$\bar{T}:=4L$}, that is $\mathcal{P} = \{ c \in \mathbb{R}^{d}_{\geq 0}: \pi \cdot c \le \bar{T} \}$. \textcolor{red}{Notice that $\pi$ and $\bar{T}$ are integer multiples of $L/\lambda^2$, and thus $\mathcal{P} = \{ c \in \mathbb{R}^{d}_{\geq 0}: \pi/(L/\lambda^2) \cdot c \le \bar{T}/(L/\lambda^2)\}$. The following lemma, that is a simple adaptation of an observation by Alon et al.~\cite{alon_approximation_1998}, shows that there exists an optimal solution for the rounded instance that uses only configurations in $\mathcal{P}$.}

\begin{lemma}
\label{lm:makespanGen}
 \textcolor{red}{For $\e>0$ small enough, the rounded instance with processing times $\bar{p}_j$ admits an optimal solution with makespan at most $4L$.}
\end{lemma}
\begin{proof}
  Among all optimal solutions to the problem, consider one that minimizes $\sum_{i}\ell_i^2$, where $\ell_i$ is the load on machine $i$. Assume that there exists a machine $i$ such that $\ell_i> 4L$. Notice that 
 \[
 \sum_j \bar{p}_j/m\le (1+\delta)^2\sum_j p'_j/m = (1+\delta)^2L'\mathop{\le}^{\text{Theorem~\ref{thm:roundAlon}}} (1+\delta)^2 (1+2/\lambda)L\le (1+\delta)^4L. 
 \]
 Since $\delta\le \e$, for $\e$ small enough ($\e\le 1/10$ suffices) we have that $(1+\delta)^4\le 2$ and thus $\sum_j \bar{p}_j/m\le 2L$. Also, recall that $\bar{p}_j\le(1+\delta)^2p_j'\le (1+\delta)^2 L\le 2L$ for any $j$, where the second to last inequality follows from Theorem~\ref{thm:roundAlon}. Since $\ell_{\min}=\min_i \ell_i\le\sum_j \bar{p}_j/m \le 2L$, then $\ell_i-\ell_{\min}> 4L - 2L = 2L.$ Then, for any job $j$, we have that $p_j< \ell_i - \ell_{\min}$. Let $j^*$ be any job assigned to machine $i$. Hence, in particular we have that $p_{j^*} < \ell_i - \ell_{\min}$.
 
 Recall that for problems (I) and (II) function $f$ is convex. Hence, it holds that $f(x+\Delta)+f(y-\Delta)\le f(x) + f(y)$ for all \(0\le x\le y \) with \(0\le \Delta \le y-x \)~\cite{alon_approximation_1998}. Moreover, the inequality becomes strict if $f$ is strictly convex. Setting $x=\ell_{\min}$, $y=\ell_i$ and $\Delta = p_{j^*}$, the inequality implies that moving job \(j^*\) to machine \(i^*\in \arg\min_i{\ell_i}\) decreases strictly $\sum_{i}\ell_i^2$. Moreover, the objective function (I) does not increase when performing this move, which yields a contradiction for this objective. Similarly, for problem (II) the objective does not increase since $\max\{f(\xi):\xi\in [x,y]\}$ is always attained at $x$ or $y$ for $f$ convex. This yields a contradiction for (II).
 
 Analogously, for problems (I') and (II') function $f$ is concave and thus $f(x+\Delta)+f(y-\Delta)\ge f(x) + f(y)$ holds for all \(0\le x\le y \) with \(0\le \Delta \le y-x \)~\cite{alon_approximation_1998}. Hence, moving job $j^*$ to machine $i^*$ decreases $\sum_{i}\ell_i^2$ but does not increase the objective (I'). Since $f$ concave implies that $\min\{f(\xi):\xi\in [x,y]\}$ is always attained at $x$ or $y$, we also obtain a contradiction for (II'). The lemma follows.   
%  Since for problems (I) and (II) (respectively (I') and (II')) function $f$ is convex (respectively concave), then moving job $j$ to $\arg\min \ell_{i}$ cannot increase (respectively decrease) the objective function (see Alon et al.~\cite{alon_approximation_1998}). Moreover, this move strictly decreases $\sum_{i}\ell_i^2$. This contradicts the choice of our solution.
\end{proof}

Let $L=\sum_j p_j/m$ be the average machine load (of the original instance). After our rounding we obtain an instance $I'$ with job set $\mathcal{J}'$ and processing times $\bar{p}_j$ for $j\in \mathcal{J}'$. Moreover, the $\bar{p}_j$ are multiples of $L/\lambda^2$, where $\lambda\ge 1/\delta$ is an integer such that $\lambda = O(1/\delta)$, and also $\bar{p}_j\ge L/\lambda$. It holds that there exists an optimal solution of the rounded instance with makespan at most $4L$, see Lemma~\ref{lm:makespanGen} (in particular $\bar{p}_j\le 4L$ for all $j$). 
Let $\Pi=\{\pi_1,\ldots,\pi_d\}$ be the distinct values that the processing times $\bar{p}_j$ can take. Our rounding guarantees that $d=|\Pi| = O((1/\delta)\log(1/\delta))$. We consider the knapsack polytope with capacity $\bar{T}:=4L$, that is $\mathcal{P} = \{ c \in \mathbb{R}^{d}_{\geq 0}: \pi \cdot c \le \bar{T} \}$. Notice that $\pi$ and $\bar{T}$ are integer multiples of $L/\lambda^2$, and that $\mathcal{P}$ can also be written as $\{ c \in \mathbb{R}^{d}_{\geq 0}: \pi/(L/\lambda^2) \cdot c \le \bar{T}/(L/\lambda^2)\}$. 

As before, we say that a configuration is simple if $|\supp{c}|\le \log (\bar{T}+1)$, and complex otherwise.  We denote by $Q_c \subseteq Q$ the set of complex configurations and by  $Q_s \subseteq Q$ the set of simple configurations. In what follows we focus on objective function (I). 

We set an ILP for the problem as before. Notice that each configuration $c$ incurs a cost of $f_c:=f(\pi \cdot c)$. Moreover, we round and scale the values $f_c$ by defining $\bar{f}_c=\lceil f_c/(\e f_{\min}) \rceil$, where $f_{\min}=\min_{c\in Q}f_c$. It is not hard to see that solving a problem with those coefficients yields a $(1+\e)$-approximate solution to the optimal solution of $I'$ with processing times $\bar{p}_j$. Let also $b_k$ be the number of jobs $j$ of processing time $\bar{p}_j=\pi_k$ \textcolor{red}{in $\mathcal{J}'$}. Consider the ILP obtained by adding to [conf-IP] the objective function $\min \sum_{c\in Q} \bar{f}_c\cdot x_c$. We call this ILP [cost-conf-IP]. With our previous discussion, it suffices to solve this ILP optimally. 
To solve this problem, we first notice that the largest coefficient in the objective can be bounded as follows.

\begin{lemma}
\label{lm:fbound}
 If $f$ satisfies Condition~\ref{cond:sensitivity} then the largest value $\max_{c\in Q}\bar{f}_c$ is upper bounded by ~$1/\delta^{O(1)}$.
\end{lemma}
\begin{proof}
 We first bound $(\max_{c\in Q} f_c) /(\e f_{\min})$. Notice that Condition~\ref{cond:sensitivity} implies that $f$ is continuous on $\mathbb{R}_{\ge0}$, and thus it admits a minimum and maximum in the interval $[L/\lambda,4L]$. Let $x_{\min}\in\arg\min \{f(x): x\in [L/\lambda,4L]\}$ and $x_{\max}\in\arg\max \{f(x): x\in [L/\lambda,4L]\}$. 
 
 Consider first the case in which $x_{\min} \le x_{\max}$ (this is not always true since $f$ might not be monotone). We now use Condition~\ref{cond:sensitivity} iteratively. Let $y^k := (1+\delta)^kx_{\min}$. Since $y^k\le y^{k-1}(1+\delta)$, Condition~\ref{cond:sensitivity} implies that $f(y^k)\le (1+\e)f(y^{k-1})$. Iterating this idea we obtain that $f(y^{k})\le (1+\e)^kf(y^0)$. Taking $k=\lceil \log_{1+\delta}(x_{\max}/x_{\min})\rceil$ implies that $x_{\max}\le y^k \le x_{\max}(1+\delta)$ and thus, by Condition~\ref{cond:sensitivity}, it holds that $f(y^k) \ge (1-\e)f(x_{\max})$. Recall that $\delta\le \e \le \gamma \delta$. We obtain that
 
 \begin{align*}
  f(x_{\max})\le \frac{f(y^{k})}{1-\e} & \le \frac{(1+\e)^{k}}{1-\e}f(x_{\min})\\
  &\le \frac{(1+\e)^{\log_{1+\delta}(x_{\max}/x_{\min})+1}}{1-\e}f(x_{\min})\\
  &\le \frac{(1+\e)^{\log_{1+\delta}(4\lambda)+1}}{1-\e}f(x_{\min})\\
  &= \frac{1+\e}{1-\e} (4\lambda)^{\log_{1+\delta}(1+\e)}f(x_{\min})\\
  &\le \frac{1+\e}{1-\e} (4\lambda)^{\log_{1+\delta}(1+\gamma\delta)}f(x_{\min})\\
  &= (1/\delta)^{O(1)} f(x_{\min}),
  \end{align*}
  
 where the last expression follows since $ \log_{1+\delta}(1+\gamma\delta)=\ln(1+\gamma\delta)/\ln(1+\delta)\le \gamma\delta/\ln(1+\delta) = O(\gamma) = O(1)$ (for $\delta$ small enough), and since $\lambda = O(1/\delta)$. We conclude that \[                                                                                                                                                                                                                                                       \max_{c\in Q}\bar{f}_c \le (f(x_{\max})) /(\e f(x_{\min})) +1 \le (1/\e)(1/\delta)^{O(1)}+1=(1/\delta)^{O(1)}.                                                                                                                                                                                                                                                      \] 
 For the case in which $x_{\max}\le x_{\min}$ we define the sequence $y^k := (1-\delta)^kx_{\min}$. The rest of the proof is analogous and the details are left to the reader.
%  Condition~\ref{cond:sensitivity} implies that since $4L' \le (1+\alpha)\cdot (L/\lambda )$ for $\alpha\in O(1/\delta)$, then $f(4L')\le (1+\beta)f(L/\lambda)$, for some $\beta = O(\alpha)$. Hence, $\max_{c\in Q} f_c) / (\delta f_{\min})\le f(4L')/(\delta f(L/\lambda))\le (1+\beta)/\lambda = O(1/\delta^2) $.
\end{proof}

As we now must consider the objective function, we cannot simply apply Theorem~\ref{thm:thin} to [cost-conf-ILP]. However, we can prove a slightly weaker version by decomposing the ILP in several smaller ones and applying the theorem to each of them.

\begin{theorem}
 \label{thm:thinGen}
 If [cost-conf-IP] is feasible, then there exists an optimal solution $x$ satisfying:
 \begin{enumerate}
%   \item $x$ is a $(1+\e)$-approximate solution with respect to the optimal cost of the ILP,
  \item $\sum_{c\in Q_c}x_c \in O((1/\delta^3)\log^2(1/\delta))$, and
  \item $|\supp{x}\cap Q_s|\in O((1/\delta)\log^2(1/\delta))$.
 \end{enumerate}
\end{theorem}
\begin{proof}
Notice that the load of each configuration $\pi \cdot c$ is a multiple of $L/\lambda^2$, and thus \textcolor{red}{$\pi \cdot c\in \{L/\lambda,L/\lambda+L/(\lambda^2),\ldots,4L\}$} \kkcom{$\bar{T}$ instead of L?}\jcom{fixed}. We classify the configurations according to their loads, $Q^\ell := \{c\in Q: \pi\cdot c = L/\lambda + \ell\cdot L/(\lambda^2)\}$, for \textcolor{red}{$\ell\in\{0,\ldots,4\lambda^2-\lambda\}$}. Let $x^*$ be an optimal solution of [cost-conf-IP]. Then we can considered an ILP for each load value $\ell$:

\begin{align}
 \text{[conf-IP]}_{\ell}\quad & \sum_{c\in Q^{\ell}} c\cdot x_c = \sum_{c\in Q^{\ell}} c\cdot x^*_c,\\
 & \sum_{c\in Q^{\ell}} x_c = \sum_{c\in Q^{\ell}} x^*_c ,\\
 &x_c \in \mathbb{Z}_{\ge0} & \text{ for all }c\in Q^{\ell}. 
\end{align}

Scaling $\pi$ by multiplying it by $\lambda^2/L$ we obtain an integral vector (since $\pi$ is an integer multiple of $L/(\lambda^2)$), we can apply Theorem~\ref{thm:thin} to each ILP [conf-IP]$_{\ell}$, which yields that there exists a thin solution $x^{\ell}$. In particular the number of complex configurations in $x^{\ell}$ is $\sum_{c\in Q_c\cap Q^{\ell}} x_c^{\ell} \in O((1/\delta)\log^2(1/\delta))$. Since $\bar{f}_c$ depends only on the load of $c$, concatenating these solutions yields a solution $x':=(x^{\ell})_{\ell}$ that is optimal for [cost-conf-IP], such that $\sum_{c\in Q_c} x_c' \in O((\lambda^2)\cdot(1/\delta)\log^2(1/\delta))=O((1/\delta^3)\log^2(1/\delta))$. It remains to bound the number of simple configurations in the support. To this end, we consider the ILP restricted to simple configurations as follows:

\begin{align}
 \nonumber
 \text{[cost-conf-IP]}_s\quad &\min \sum_{c\in Q_s} \bar{f}_c\cdot x_c\\ 
 & \sum_{c\in Q_s} c\cdot x_c = b-\sum_{c\in Q_c} c\cdot x'_c,\\
 & \sum_{c\in Q_s} x_c = m-\sum_{c\in Q_c} x'_c,\\
 &x_c \in \mathbb{Z}_{\ge0} & \text{ for all }c\in Q_s. 
\end{align}

We apply the result of Eisenbrand and Shmonin~\cite{eisenbrand_caratheodory_2006} to this ILP. In its more general form, this result ensures the existence of a solution $x''$ with support of size $O(N(\log(N) + \Delta))$, where $N$ is the number of restrictions and $\Delta$ is the encoding size of the largest coefficient appearing in the cost vector and restriction matrix. In our case $N= d+1 = O((1/\delta)\log(1/\delta))$, and $\Delta = O(\log(\max\{1/\delta, \max_{c\in Q} \bar{f}_c)\} )=O(\log(1/\delta))$ (Lemma~\ref{lm:fbound}). Thus $O(N(\log(N) + \Delta))= O((1/\delta)\log^2(1/\delta))$. The theorem follows by concatenating $(x''_c)_{c\in Q_s}$ with $(x'_c)_{c\in Q_c}$.
\end{proof}

Finally, we use the structure given by the theorem to solve this ILP optimally.

\begin{algorithm}
\label{alg:PTASparallelGen}
\hspace{4cm}\begin{enumerate}
\item For each processing time $\pi_k$, guess the number $b^c_k \leq b_k$ of jobs covered by complex configurations.% $Q_c \subseteq Q$.
\item Guess the number $m^c$ of machines that schedule jobs $b^c$.
\item Compute an optimal solution for instance with number of jobs $b^c$ on $m^c$ machines with a dynamic program.
\item Guess the support of simple configurations ${\bar Q}_s \subseteq Q_s$ used by the solution implied by Theorem~\ref{thm:thin}, with $|{\bar Q}_s| \in O((1/\delta)\log^2 (1/\delta))$.
\item Solve the ILP restricted to configurations in $\bar{Q}_s$:
\begin{align*}
 & \min \sum_{c\in Q_s} \bar{f}_c\cdot x_c\\
 & \sum_{c\in \bar{Q}_s} c\cdot x_c =b-b^c,\\
 & \sum_{c\in {\bar Q}_s} x_c = m-m^c,\\
 &x_c \in \mathbb{Z}_{\ge 0} & \text{ for all }c\in \bar{Q}_s. 
\end{align*}
\end{enumerate}
\end{algorithm}

\begin{lemma} \label{lm:runningtimeGen}
Algorithm~\ref{alg:PTASparallelGen} can be implemented with a running time of $2^{O((1/\delta)\log^4(1/\delta))} \log(n)$.
\end{lemma}
\begin{proof}
In step 1, the algorithm guesses which jobs are processed on machines following a complex configurations. Since each configuration contains at most $O(1/\delta)$ jobs, there are at most $O((1/\delta^4)\log^2 (1/\delta))$ jobs assigned to such machines. For each size $\pi_k \in \Pi$, we guess the number $b_k^c$ of jobs of size $\pi_k$ assigned to such machines. Hence, we can enumerate all possibilities for jobs assigned to complex machines in time $2^{O((1/\delta)\log^2 (1/\e))}$. Similarly, we can guess the number of machines in step 2 since $m_c\in O((1/\delta^3)\log^2 (1/\delta))$. For step 3 we use a simple dynamic program that goes over the machines storing a table $T(\ell,z_1,\ldots,z_d)$ that contains the minimum cost achieved over the first $\ell\le m^c$ machines with $z_k\le b_k^c$ jobs of size $\pi_k$. 
\textcolor{red}{The number of entries the table is $O(m^c\prod_{k=1}^d (b_k^c+1))$. Computing $T(\ell,z_1,\ldots,z_d)$ can be done by checking all entries of the type $T(\ell-1,z_1',\ldots,z_d')$ for $z_k'\le z_k$. Thus, the running time of the dynamic programm is $O(m^c [\prod_{k=1}^d (b_k^c+1)]^2)$.} Since $b_k^c \in O((1/\delta^4)\log^2 (1/\delta))$ for each $k$, recalling that $m^c\in O((1/\delta^3)\log^2 (1/\delta))$, and that $d=|\Pi| \in O((1/\delta) \log(1/\delta))$, we obtain that step 3 can be implemented with $2^{O((1/\delta)\log^2 (1/\delta))}$ running time.

In step 4, our algorithm guesses the support of the solution implied by Theorem~\ref{thm:thinGen}. Let $D$ be the bound implied by the third property of this theorem, so that $|\supp{x}\cap Q_s|\le D$ and $D\in O((1/\delta) \log^2(1/\delta))$. Also, $|Q_s| \le 2^{O(\log^2 (1/\delta))}$ Then the guessing in step 4 needs to consider the following number of possibilities:

\[
\sum_{i=0}^{D} {|Q_s| \choose i} \le (D+1)  |Q_s|^{D} \le 2^{O((1/\delta)\log^4 (1/\delta))}.
\]

In step 5, the number of variables of the restricted ILP is $|\bar{Q}_s| = O((1/\delta)\log^2 (1/\delta))$. Moreover, using Lemma~\ref{lm:fbound} the size of the input can be bounded by $O((1/\delta^2)\log^3(1/\delta)\log(n))$. Running Kannan's algorithm \cite{kannan_minkowskis_1987} to solve the ILP takes time $2^{O((1/\delta)\log^3(1/\delta))}\log(n)$. Hence, the total running time of our algorithm can be bounded by $2^{O((1/\delta)\log^4(1/\delta))}\log(n)$. 
\end{proof}

As in~\cite{alon_approximation_1998}, the algorithm above can be easily adapted for objectives (II), (I') and (III') by suitably adapting the ILP. We leave the details to the reader. This suffices to conclude Theorem~\ref{thm:PTASGen}

\begin{theorem}
\label{thm:PTASGen}
 Consider the scheduling problem on parallel machines with objective functions (I), (II) for $f$ convex (respectively (I') and (II') for $f$ concave). If $f$ satisfies Condition~\ref{cond:sensitivity} for $1/\delta = O(1/\e)$, then the problem admits an EPTAS with running time $2^{O((1/\e)\log^4(1/\e))}+O(n\log n)$.
\end{theorem}

\section{Minimum makespan scheduling on uniform machines}

In this section we generalize our result for $P||C_{\max}$ to uniform machines. Consider a set of jobs ${\cal J}$ with
processing times $p_j$ and a set of $m$ non-identical machines $\mathcal{M}$
where machine $i\in\mathcal{M}$ runs at speed $s_i$. If job $j$ is executed on
machine $i$ the machine needs $p_j/s_i$ time units to complete the
job. The problem is to find an assignment $a: {\cal J} \to {\cal M}$
for the jobs to the machines that minimizes the makespan;
$\max_i \sum_{j : a(j) = i} p_j/s_i$. The problem is denoted by $Q||C_{max}$. We suppose that $s_1 \ge s_2 \ge \ldots \ge s_m$.
Jansen~\cite{jansen_eptas_2010} found an efficient polynomial time approximation scheme (EPTAS) for this scheduling problem which has a running time of $2^{O(1/\e^2 \log^3(1/\e))} + \text{poly}(n)$. Here we show how to improve the running time and prove the main result of this section.

\begin{theorem} \label{thm:QCmax}
There is an EPTAS (a family of algorithms $\{A_\e : \e > 0\}$) which, given an instance $I$ of $Q||C_{max}$ with $n$ jobs and $m$ machines and a positive number $\e > 0$, produces a schedule of makespan $A_\e(I) \le (1+\e) \OPT(I)$. The running time of $A_\e$ is $2^{O(1/\e \log^4(1/\e))} + \textnormal{poly}(n)$.
\end{theorem}

We follow the approach by Jansen~\cite{jansen_eptas_2010}, transforming the scheduling problem into a bin packing problem with different bin capacities, round the processing times and bin capacities, divide the bins into at most three groups and generate four different scenarios depending on the input instance.

First, we compute a $2$-approximate solution using the algorithm by Gonzales et al. \cite{gonzales} of length $B(I) \le 2 \, \OPT(I)$. Suppose that $\e < 1$; otherwise we can take the $2$-approximate solution and are done. Then we choose a value $\delta \in (0,\e)$ such that $1/\delta$ is integral (the exact value is specified later) and use a binary search within the interval $[B(I)/2,B(I)]$ (that contains $\OPT(I)$). We use a standard dual approximation method that for each value $T$ either computes an approximate schedule of length $T(1+ a \delta)$ (where $a$ is constant) or shows that there is no schedule of length $T$.
Since $(\delta/2) B(I) \le \delta \OPT(I)$, we can find within $O(\log(1/\delta))$ iterations a value $T \le \OPT(I)(1+\delta)$ with a corresponding schedule of length at most $T(1+a\delta) \le \OPT(I) (1+\e)$, using $\delta \le \e/(a+2)$ and $\e \le 1$. Next, the scheduling problem is transformed into a bin packing problem with $m$ bins and capacities $c_i = T \cdot s_i$, the processing times $p_j$ are rounded to the next value $\bar{p}_j$ of the form $\delta (1+\delta)^{k_j}$ with $k_j \in \gz$ and the bin capacities are rounded to the next power of $(1+\delta)$. We call $\mathcal{B}$ the set of bins with rounded capacities.

\begin{lemma}[Jansen \cite{jansen_eptas_2010}]
If there is a feasible packing of $n$ jobs with processing times $p_j$
into $m$ bins with capacities $c_i$, then there is also a packing of $n$ jobs with rounded processing times $\bar{p}_j = \delta (1+\delta)^{k_j} \le (1+\delta)p_j$ into $m$ bins with rounded bin capacities $c'_i = (1+\delta)^{\textcolor{red}{\ell_i}} \le c_i (1+\delta)^2$ with $\ell_i \in \mathbb{Z}$.%\jcom{What is $\ell_i$?}
\end{lemma}

If the number $m$ of bins is smaller than $K \in O(1/\delta \log(1/\delta))$,
then we can use an approximation scheme by Jansen and Mastrolilli \cite{monaldo} to compute an $(1+\e)$-approximate solution to schedule $n$ jobs on $m$ unrelated machines (an even more general problem) within time $O(n) + 2^{O(m \log(m/\e))} = O(n) + 2^{O(1/\delta \log^2(1/\delta))} = O(n) + 2^{O(1/\e \log^2(1/\e))}$; using that $\delta \in O(\e)$. Suppose from now on that $K > O(1/\delta \log(1/\delta))$. Then, we divide the bins into at most three different bin groups.
The first group ${\cal B}_1$ consists of the $K = O(1/\delta \log(1/\delta))$ largest bins. For some $\gamma \in \Theta(\e^2)$, the next group ${\cal B}_2$ consists either of all the remaining bins $\{b_{K+1},\ldots,b_m\}$ if $c'_m > \gamma c'_K$ (and we have only two bin groups)
or ${\cal B}_2$ contains the next $G$ largest bins $\{b_{K+1},
\ldots,b_{K+G}\}$ where $G$ is the smallest index such that capacity $c'_{K+G+1} \le \gamma c'_K$.
In the second case, ${\cal B}_3 = \{b_{K+G+1},\ldots,b_m\}$. Let $c_{max}({\cal B})$ and $c_{min}({\cal B})$ be the largest and smallest bin capacity in ${\cal B}$. If $c_{max}({\cal B})/c_{min}({\cal B}) \le C$ for some value $C$ and ${\cal B}$ contains only rounded capacities $(1+\delta)^x$ with $x \in \gz$, then the number of different capacities
in ${\cal B}$ is at most $O(1/\delta \log(C))$.

\begin{lemma}[Jansen~\cite{jansen_eptas_2010}]
If there is a solution for the original instance $({\cal J}, {\cal M})$
of our scheduling problem with makespan $T$ and corresponding bin sizes, then there is a feasible packing for instance $({\cal J},{\cal B}'_1 \cup {\cal B}_2 \cup {\cal B}_3)$ or instance $({\cal J},{\cal B}'_1 \cup {\cal B}_2)$ with rounded bin capacities $\bar{c}_i \le c_i(1+\delta)^3$
and rounded processing times $\bar{p}_j \le (1+\delta) p_j$. Here ${\cal B}'_1$ is the subset of ${\cal B}_1$ with bins of capacity larger than $\delta/(K-1) c_{max}({\cal B}_1)$ and ${\cal B}_2$ has a constant number $O(1/\delta \log(1/\delta))$ of different bin capacities. In addition we have one of the following four scenarios:

\begin{description}
\item[(1)] Two bin groups ${\cal B}'_1$ and ${\cal B}_2$ with a gap $c_{min}({\cal B}'_1)/c_{max}({\cal B}_2)  \ge 1/\delta$.

\item[(2)] Two bin groups ${\cal B}_1$ and ${\cal B}_2$ with a constant number $O(1/\delta \log(1/\delta))$ of different bin capacities
    in ${\cal B}'_1 \cup {\cal B}_2$.

\item[(3)] Three bin groups ${\cal B}'_1,{\cal B}_2,{\cal B}_3$ with a gap
    $c_{min}({\cal B}'_1)/c_{max}({\cal B}_2) \ge 1/\delta$ and
    $c_{min}({\cal B}_1)/c_{max}({\cal B}_3) \ge 1/\gamma$.

\item[(4)] Three bin groups ${\cal B}'_1,{\cal B}_2,{\cal B}_3$ with a  constant number $O(1/\delta \log(1/\delta))$ of different bin capacities in ${\cal B}'_1 \cup {\cal B}_2$ and a gap
     $c_{min}({\cal B}_1)/c_{max}({\cal B}_3) \ge 1/\gamma$.
\end{description}
\end{lemma}

Notice that scenario $4$ can be seen as a special case of scenario $3$
by using ${\cal B}_2^{new} = {\cal B}_1 \cup {\cal B}_2$, ${\cal B}_1^{new} = \emptyset$, and ${\cal B}_3^{new} = {\cal B}_3$. The same modification works to show that scenario $2$ is a special case of scenario $1$. Finally, scenario $1$ can be interpreted a special case
of scenario $3$, using ${\cal B}_3 = \emptyset$. Therefore, it is sufficient to improve the running time for scenario $3$.

Scenario $3$ will be solved using a mix of dynamic programming and mixed integer linear programming (MILP) techniques. In this approach we use only the larger bins in ${\cal B}'_1$ of $\mathcal{B}_1$ to execute jobs, but in the rounding step for scenario $3$ afterwards we may also use the smaller bins in ${\cal B}_1$. Notice that a packing into a bin $b_i$ with capacity $\bar{c}_i \le c_i(1+\delta)^3$ corresponds to a schedule on machine $i$ with total processing time at most $\bar{c}_i/s_i \le
T(1+\delta)^3$. For $T \le \OPT(I)(1+\delta)$ this gives us a schedule of length at most $\OPT(I) (1+\delta)^4$.  If there a feasible schedule with makespan $T$, then the total processing time of the instance is $\sum_{j \in {\cal J}} p_j \le \sum_{i=1}^m \bar{c}_i$. If this inequality does not hold, then we discard the choice with makespan $T$. Otherwise, we can eliminate the set ${\cal J}_{\text{tiny}}$ of tiny jobs with processing time $\le \delta \bar{c}_m$ and pack them greedily at the end of the algorithm into the enlarged bins of size $\bar{c}_i (1+\delta)$. Hence, in what follows we assume that ${\cal J}_{\text{tiny}}$ is empty.

\subsection{Solution for the instance $({\cal J},{\cal B}'_1 \cup {\cal B}_2 \cup {\cal B}_3)$}

In this subsection we consider scenarios $3$ above with three bin groups.
First, we preassign all huge jobs
with processing time $> \delta \bar{c}_{K'}$ into the first $K'$ machines.
Since $\bar{c}_{K+1} = c_{max}({\cal B}_2)  \le \delta c_{min} ({\cal B}'_1) = \delta \bar{c}_{K'}$, the huge jobs fit only on the first $K'$ bins.
The number of huge jobs can be bounded by $O(K c_{max}({\cal B}_1) / (\delta c_{K'})) = O(1/\delta^4 \log^2(1/\delta))$. If there are more huge jobs in the instance, then there is no packing into ${\cal B}'_1 \cup {\cal B}_2$ and we are done.
Furthermore, the number of machines $K' \in O(1/\delta \log(1/\delta))$ is constant. Again, we can use the approximation scheme by Jansen and Mastrolilli 
%\cite{monaldo} \jcom{Reference is missing} 
that computes an $(1+\delta)$-approximate schedule for $N$ jobs on $M$ machines which runs in $O(N) + 2^{O(M \log(M/\delta))}$. For $M \in O(1/\delta \log(1/\delta))$ and $N \in O(1/\delta^4 \log^2(1/\delta))$ this gives a running time
$O(1/\delta^4 \log^2(1/\delta)) + 2^{O(1/\delta \log^2(1/\delta))} = 2^{O(1/\delta \log^2(1/\delta))}$ to obtain a feasible packing with bin sizes $\bar{c}_i (1+\delta)$ or schedule of length $\le T(1+\delta)^4$, if one exists.
If there is no feasible packing for the huge jobs, then there is no schedule with makespan $T$ and we have to increase $T$ in the binary search. In the other case we set up a MILP.

After the assignment of the huge jobs, we have a free area $S_0$ in ${\cal B}_1$  for the remaining jobs with processing time $\bar{p}_j \le \delta \bar{c}_{K'}$. The different bin capacities in ${\cal B}_2$ and ${\cal B}_3$ are denoted by $\bar{c}(1) > \ldots > \bar{c}(L)$ and $\bar{c}(L+1) > \ldots > \bar{c}(L+N)$, respectively. Let $m_\ell$ be the number of bins of size $\bar{c}(\ell)$ for $\ell = 1,\ldots,L+N$. The $m_\ell$ machines of the same speed form a block $B_\ell$ of bins with the same capacity $\bar{c}(\ell)$. In addition, we have $n_1,\ldots,n_P$ jobs of size $\delta (1+\delta)^{k_j}$ and suppose that the first $P' \le P$ job sizes are larger than $\bar{c}_{K+1} = \bar{c}(1)$.

In the MILP we use $C_1^{(\ell)}, \ldots,C_{h_\ell}^{(\ell)}$ as configurations or multisets with numbers $\delta (1+\delta)^{k_j} \in [\delta \bar{c}(\ell),\bar{c}(\ell)]$ (these are large processing times corresponding to block $B_\ell$), where the total sum $size(C_i^{(\ell)}) = \sum_j a(k_j,C_i^{(\ell)}) \delta (1+\delta)^{k_j}$ is bounded by $\bar{c}(\ell)$. Here $a(k_j,C_i^{(\ell)})$ is the number of occurrences of number $\delta (1+\delta)^{k_j}$ in configuration $C_i^{(\ell)}$. In the MILP below, we use integral and fractional variables $x_i^{(\ell)}$ to indicate number of machines that are scheduled according to configuration $C_i^{(\ell)}$. In addition, we use fractional variables $y_{j,\ell}$ to indicate the number of jobs of size $\delta(1+\delta)^{k_j}$ placed as small ones in block $B_\ell$; i.e. $\delta (1+\delta)^{k_j} < \delta \bar{c}(\ell)$.
For each job size $\delta (1+ \delta)^{k_j} \le \bar{c}(1)$, let $a_j$ be the
smallest index in $\{1,\ldots,L+N\}$ such that $\delta (1+\delta)^{k_j} \ge
\delta \bar{c}(a_j)$. If there is no such index, we have a tiny processing
time $\delta (1+\delta)^{k_j} < \delta \bar{c}(L+N)$. Notice that the first $P'$
job sizes are within $( {\bar c}(1),\delta c_{K'} ]$. These jobs do not fit into
${\cal B}_2 \cup {\cal B}_3$. Therefore, for these job sizes we use only one
variable $y_{j,0} = n_j$ and set $a_j = 0$.

$$ \begin{array}{llll}
 \sum_{i} x_i^{(\ell)} & \le m_\ell & {\rm for \ } & \ell=1,\ldots,L+N, \\
 \sum_{\ell,i} a(k_j,C_i^{(\ell)}) x_i^{(\ell)} + \sum_{\ell=0}^{{a_j}-1} y_{j,\ell} & = n_j & {\rm for \ } & j=P'+1,\ldots,P ,\\
 \sum_{i} size(C_i^{(\ell)})x_i^{(\ell)} + \sum_{j:\ell < a_j} y_{j,\ell} \delta(1+\delta)^{k_j} & \le m_\ell \bar{c}(\ell) & {\rm for \ } & \ell=1,\ldots,L+N, \\
 \sum_{j=1}^P y_{j,0} \delta (1+\delta)^{k_j} & \le S_0, & & \\
 \end{array} $$

$$ \begin{array}{lll}
 x_i^{(\ell)} \, \text{integral} \ge 0 & {\qquad } &  {\rm for \ } \ell=1,\ldots,L
           {\rm \ and \ } i=1,\ldots,h_\ell, \\
 x_i^{(\ell)} \ge 0  & & {\rm for \ } \ell=L+1,\ldots,L+N {\rm \ and \ }
                          i=1,\ldots,h_\ell \\
 y_{j,0} = n_j    & {\qquad } &  {\rm for \ } j=1,\ldots,P', \\
 y_{j,\ell} \, \text{integral} \ge 0  & {\qquad } &  {\rm for \ } j=P'+1,\ldots,P
   {\rm \ and \ }  \ell=0,\ldots,a_j-1.\\
 \end{array} $$%\jcom{$a_j$ in the MILP has not been defined}

In the MILP above, we use integral variables for configurations in the blocks of group ${\cal B}_2$ and fractional variables for the configurations in blocks of ${\cal B}_3$.  Each feasible packing for the jobs into the bins corresponds to a feasible solution of the MILP. The total number of variables is $O(n^2) + O(n) 2^{O(1/\delta \log(1/\delta))}$, the number of integral variables is at most $2^{O(1/\delta \log(1/\delta))}$, and the number of constraints (not counting the non-negativity constraints) is at most $O(n)$.
The previous approach to solve the scheduling problem and the underlying MILP had a running time of $2^{O(1/\delta^2 \log^3(1/\delta))} + \text{poly}(n)$.
In order to use an approach similar to the scheduling on identical machines, each large size $\delta (1+\delta)^{k_j} \in C_i^{(\ell)}$ is rounded up to the next multiple of $\delta^2 \bar{c}(\ell)$. This enlarges the size of each configuration $C_i^{(\ell)}$ from $size(C_i^{(\ell)})$ to at most $size(C_i^{(\ell)}) + \delta \bar{c}(\ell)$ and the corresponding bin size from $\bar{c}(\ell)$ to $(1+\delta) \bar{c}(\ell)$.

Let $\bar{C}_1^{(\ell)},\ldots, \bar{C}_{\bar{h}_\ell}^{(\ell)}$ be the configurations of size at most $(1+\delta) \bar{c}(\ell)$ with the rounded-up numbers $q(k_j,\ell) \delta^2 \bar{c}(\ell)$ with $q(k_j,\ell) \in \gz^+$ and multiplicities $a(k_j,\bar{C}_i^{(\ell)})$. This rounding implies also that the rounded size $size(\bar{C}_i^{(\ell)})$ of a configuration is a multiple of $\delta^2 \bar{c}{(\ell)}$.
Each new rounded configuration $\bar{C}_i^{(\ell)}$ (with rounded-up numbers $q(k_j,\ell) \delta^2 \bar{c}(\ell)$ and multiplicities $a(k_j,\bar{C}_i^{(\ell)})$) corresponds to an integral point inside the knapsack
polytope $\mathcal{P}_\ell = \{C = (a(k_j,C)) \,:\, q \cdot C \le 1/\delta^2 + 1/\delta\}$ 
%\jcom{I'm confused, shouldn't it $C = (a(k_j,C))$ instead of $C = (q(k_j,C))$ in the def. of $\mathcal{P}_\ell $. Also, the right hand of $q \cdot C \le 2/\delta^2$ should be left as $1/\delta^2 + 1/\delta$, o.w. the knapsack polytope does not define exactly the mentioned configurations.}  
such that $\sum_j  q(k_j,\ell) a(k_j,\bar{C}_i^{(\ell)}) \delta^2 \bar{c}(\ell) =
size(\bar{C}_i^{(\ell)}) \le (1+\delta) \bar{c}(\ell)$
or, equivalently,
$\sum_j q(k_j,\ell) a(k_j,\bar{C}_i^{(\ell)})  \le 1/\delta^2 + 1/\delta \le 2/\delta^2$. We consider now a modified MILP with configurations $\bar{C}_i^{(\ell)}$ and coefficients $a(k_j,\bar{C}_i^{(\ell)})$.
Note that the total area of all configurations in $B_\ell$ can be bounded by
$\sum_i size(\bar{C}_i^{(\ell)})  x_i^{(\ell)}
 \le \sum_i size(C_i^{(\ell)})  x_i^{(\ell)} + \delta \bar{c}(\ell) \sum_i x_i^{(\ell)}$.
This, together with the small jobs gives
$ \sum_{i} size(\bar{C}_i^{(\ell)}) x_i^{(\ell)} + \sum_j \delta (1+\delta)^{k_j} y_{j,\ell} \le \sum_{i} size(C_i^{(\ell)}) x_i^{(\ell)} + \delta m_\ell \bar{c}(\ell) +  \sum_j \delta (1+\delta)^{k_j}
y_{j,\ell} \le m_\ell \bar{c}(\ell) (1 +\delta);$
i.e. the total area is increased by at most a multiplicative
factor of $(1+\delta)$.
Since the total area of all jobs within one block is increased by this rounding, we use the following new constraints in the modified MILP:
\begin{align*}
\sum_i size(\bar{C}_i^{(\ell)}) x_i^{(\ell)} + \sum_j y_{j,\ell} \delta (1+\delta)^{k_j} \le m_\ell \bar{c}(\ell) (1+\delta) & \qquad {\rm for \ } \ell=1,\ldots,L.
\end{align*}
Next, we divide the coefficients in the $L$
area constraints above by $\delta^2 \bar{c}(\ell)$. Then the coefficients
of the $x_i^{(\ell)}$ variables are now $size(\bar{C}_i^{(\ell)})/(\delta^2 \bar{c}(\ell))
= a_{i,\ell} \delta^2 \bar{c}(\ell) /(\delta^2 \bar{c}(\ell)) = a_{i,\ell}
\in \{1/\delta,\ldots,1/\delta^2 + 1/\delta\}$. Using the assumption that $1/\delta$ is integral, all coefficients of the variables are integral and bounded by $2/\delta^2$. Notice that increasing the capacities of all bins and dividing all coefficients as above, implies also a feasible solution of the modified MILP.
Let us study a feasible solution of the modified MILP.
To reduce the number of integral configuration variables in the MILP, we consider the following ILP that uses only the integral $x_i^{(\ell)}$ variables within bin group ${\cal B}_2$:

 \begin{align}
 \label{eq:ILP1}
 & \sum_{i} x_i^{(\ell)} = \bar{m}_\ell & {\rm for \ } & \ell=1,\ldots,L, \\
 \label{eq:ILP2}
 & \sum_{\ell,i} a(k_j,\bar{C}_i^{(\ell)}) x_i^{(\ell)}  = \bar{n}_j & {\rm for \ } & j \in P({\cal B}_2) ,\\
 \label{eq:ILP3}
 & \sum_{i} \frac{size(\bar{C}_i^{(\ell)})}{\delta^2 \bar{c}(\ell)}  x_i^{(\ell)} = Area(\ell,large) & {\rm for \ } & \ell=1,\ldots,L, \\
 \label{eq:ILP4}
 & x_i^{(\ell)} \, \text{integral}  \ge 0 &   {\rm for \ } & i=1,\ldots,\bar{h}_\ell,\ell=1,\ldots,L.
 \end{align} 

where the values $\bar{m}_\ell$, $\bar{n}_j$, and $Area(\ell,large)$ are
given by a feasible solution of the modified MILP. Here $P({\cal B}_2)$ is the set of all indices of large job sizes corresponding to blocks $B_\ell \in {\cal B}_2$; i.e. $P({\cal B}_2) = \{j \,:\, \delta(1+\delta)^{k_j} \in (\delta \bar{c}(L),\bar{c}(1)]\}$. The cardinality of $P({\cal B}_2)$ and the value $L$ can be
bounded by $O(1/\delta \log(1/\delta))$. All the coefficients above of the variables are bounded by $O(1/\delta^2)$.

The support of a configuration $\bar{C}_i^{(\ell)}$ is the number of values $a(k_j,\bar{C}_i^{(\ell)}) > 0$; i.e.
$\supp{\bar{C}_i^{(\ell)}} = |\{j\, :\, a(k_j,\bar{C}_i^{(\ell)}) > 0\}|$. In our case $\supp{\bar{C}_i^{(\ell)})} \le O(1/\delta \log(1/\delta))$.  A configuration $\bar{C}_i^{(\ell)}$ is called simple, if $|\supp{\bar{C}_i^{(\ell)}}| \le \log(1/\delta^2 + 1 / \delta + 1)$ 
%\jcom{Here it should be $|\supp{\bar{C}_i^{(\ell)}}| \le \log(1/\delta^2 + 1/\delta + 1)$}. 
Otherwise, we call a configuration $\bar{C}_i^{(\ell)}$ complex. Using the result by Eisenbrand and Shmonin, we can find a feasible solution of the ILP above (if there is a feasible solution of the modified MILP) with at most $O(1/\delta \log^2(1/\delta))$ many variables $x_i^{(\ell)} > 0$; i.e. $|\supp{x}| \le O(1/\delta \log^2(1/\delta))$ where $x = (x_i^{(\ell)})$. We can generalize our result  in Theorem~\ref{thm:thin} to our ILP above.

\begin{lemma} \label{thin:solution}
Assume that the ILP defined by \eqref{eq:ILP1}-\eqref{eq:ILP4} is feasible and let $S$ denote the set of all simple configurations. Then there exists a feasible solution $x'$ such that:
\begin{description}
\item[(1)] If ${x'}_i^{(\ell)} > 1$ then the configuration $\bar{C}_i^{(\ell)}$ is simple.

\item[(2)] The support of $x'$ satisfies $|\supp{x'} \cap S| \in O(1/\delta \log^2(1/\delta))$.

\item[(3)]The support of $x'$ satisfies $|\supp{x'} \setminus S| \in O(1/\delta^2 \log^3(1/\delta))$.

\end{description}
\end{lemma}
\begin{proof}
As stated above, the set of configurations $\bar{C}_1^{(\ell)},\ldots, \bar{C}_{\bar{h}_\ell}^{(\ell)}$ equals the set of integral points $Q_{\ell}$ inside the knapsack polytope $\mathcal{P}_\ell = \{C = (q(k_j,C)) \,:\, q \cdot C \le 1/\delta^2 + 1/ \delta \}$ %\jcom{This should be adjusted if the def of $\mathcal{P}_\ell$ is changed above}. 
Let $\bar{x}=(\bar{x}^{(\ell)})_{\ell=1}^{L}$ be a solution to
\eqref{eq:ILP1}-\eqref{eq:ILP4} where $\bar{x}^{(\ell)}$ corresponds to the variables defining the solution for block $B_{\ell}$. We consider a family of ILPs defined for each $\ell=1,\ldots,L$.%\jcom{Please check the ILP }
\begin{align*}
 \text{[conf-IP]}_{\ell}\quad  & \sum_{c\in Q_{\ell}} c\cdot x_c = \sum_{c\in Q_{\ell}} c\cdot \bar{x}^{(\ell)}_c,\\
 & \sum_{c\in Q_{\ell}} x_c = \bar{m}_{\ell},\\
 &x_c \in \mathbb{Z}_{\ge0} & \text{ for all }c\in Q_{\ell}. 
\end{align*}
Using Theorem \ref{thm:thin} for each [conf-IP]$_{\ell}$, we obtain new solution $\hat{x}^{(\ell)}$, where each complex configuration is used at most once and $\supp{\hat{x}^{(\ell)}} \in O(1/\delta \log^2(1/\delta))$. Then we define a new solution $\hat{x}$ of ILP \eqref{eq:ILP1}-\eqref{eq:ILP4} defined as $(\hat{x}^{(\ell)})_{\ell}$. In $\hat{x}$ every complex configuration is used at most once and $|\supp{\hat{x}}| \leq L \cdot 2(\bar{d}+1)\log(4(\bar{d}+1)\bar{T}) \in O(1/\delta^2 \log^3(1/\delta))$, where $\bar{d}\le |P(\mathcal{B}_2)|\in O(1/\delta\log 1/\delta)$  and $\bar{T}\in O(1/\delta^2)$. Note that Equation~\eqref{eq:ILP3} of the above ILP holds for the new solution $\hat{x}$ as the set of jobs covered inside a block does not change and hence 
\begin{align*}
\sum_{i} \frac{size(\bar{C}_i^{(\ell)})}{\delta^2 \bar{c}(\ell)}  {\hat{x}_i}^{(\ell)} & = Area(\ell,large) = \sum_{i} \frac{size(\bar{C}_i^{(\ell)})}{\delta^2 \bar{c}(\ell)}  {\bar{x}}_{i}^{(\ell)}.\qedhere
%Area(\ell,large) = \sum_{i} \frac{size(\bar{C}_i^{(\ell)})}{\delta^2 \bar{c}(\ell)}  x_i^{(\ell)} = \sum_{i}   \bar{s}_i \sum_{\ell,k} a(k_j,\bar{C}_k^{(\ell)}) x_k^{(\ell)} = \sum_{i} 
\end{align*}

Finally, consider the ILP \eqref{eq:ILP1}-\eqref{eq:ILP4} and fix each variable $x_i^{\ell}$, for $\bar{C}_i^{(\ell)}$ a complex configuration, to the value $\hat{x}_i^{\ell}$ (and thus the resulting ILP has variables only for simple configurations). Now we can apply the result of Eisenbrand and Shmonin~\cite{eisenbrand_caratheodory_2006} to this ILP. This ensures that any ILP of the form $\{z\in \mathbb{Z}_{\ge 0}: Az=h\}$ admits a solution with support of size $O(N(\log(N) + \Delta))$, where $N$ is the number of rows of $A$ and $\Delta$ is the largest encoding size of an entry of $A$. Recalling that $\frac{size(\bar{C}_i^{(\ell)})}{\delta^2 \bar{c}(\ell)} \in O(1/\delta^2)$, we can apply this result to our case, which yields a solution whose support contains at most $O(1/\delta \log^2(1/\delta))$ simple configurations. Hence, we obtain a solution satisfying all properties of the statement of the theorem.
\end{proof}

\begin{algorithm}
\label{alg:PTASQCmax}
\hspace{4cm}\begin{enumerate}
\item For each job size, guess the number of jobs $v_{j} \leq \bar{n}_j$ covered by complex configurations.
\item For each bin size, guess the number of machines $w_j \leq \bar{m}_j$ used to schedule the set of jobs covered by complex configurations.
\item For each block $\ell$ in $\mathcal{B}_2$, guess the support of simple configurations ${\bar Q}_{s}^{(\ell)} \subseteq Q_{s}^{(\ell)}$ used by a thin solution, with $\sum_{\ell = 1}^L |{\bar Q}_{s}^{(\ell)}| \leq 4(d+1)\log(4(d+1)\bar{T})\in O((1/\e)\log^2 (1/\e))$.
\item Solve the reduced modified MILP, where the integral variables $x_{i}^{\ell}$ are restricted to simple configurations.
%\begin{align*}
% & \sum_{c\in \bar{Q}_s} c\cdot x_c =b-b^c,\\
% & \sum_{c\in {\bar Q}_s} x_c = m-m^c,\\
% &x_c \in \mathbb{Z}_{\ge 0} & \text{ for all }c\in \bar{Q}_s. 
%\end{align*}
\end{enumerate}
\end{algorithm}

\begin{lemma} \label{lm:runningtimeQ}
Algorithm~\ref{alg:PTASQCmax} can be implemented with a running time of $2^{O((1/\e)\log^4(1/\e))} poly(n)$.
\end{lemma}
\begin{proof}
As in the case of identical machines, our algorithm guesses in step 1 the complex configurations and the corresponding jobs. Since the number $M$ of complex configurations within ${\cal B}_2$ ist at most $ O(1/\delta^2 \log^3(1/\delta))$ and there are at most $O(1/\delta)$ many large job per configuration, the total number $N$ of jobs within the complex configurations is at most $O(1/\delta^3 \log^3(1/\delta))$.

To obtain a schedule for the guessed jobs, notice that the number of large job sizes $|P({\cal B}_2)| \le O(1/\delta \log(1/\delta))$. We guess now a vector $v = (v_j)$ with possible job sizes that are covered by the complex configurations. The total number of these vectors is $(N+1)^{|P({\cal B}_2)|} \le (1/\delta^3 \log^3(1/\delta))^{O(1/\delta \log(1/\delta))} = 2^{O(1/\delta \log^2(1/\delta))}$. In addition, we guess a vector $w = (w_\ell)$ with the numbers $w_\ell$ of complex configurations in the block groups $B_\ell$. The number of choices here is at most $(M+1)^L \le (1/\delta^2 \log^3(1/\delta))^{O(1/\delta \log(1/\delta))} = 2^{O(1/\delta \log^2(1/\delta))}$. For each guess $v,w$ we run a dynamic program to test whether the number of job sizes, stored in $v$, fit on the corresponding machines in the blocks $B_\ell$, given by vector $w$. To do this, we run over the machines and store after $\ell$ machines, for $\ell=1,\ldots,M$, the set of all  feasible vectors with job sizes that can be packed into the first $\ell$ machines. This dynamic program
runs in time $M 2^{O(1/\delta \log^2(1/\delta))} = 2^{O(1/\delta \log^2(1/\delta))}$. For each feasible choice of $v,w$ we compute the reduced MILP by $\hat{m}_\ell = m_\ell - w_\ell$ and $\hat{n}_j = n_j - v_j$ and guess the support of a feasible solution $x$ in the MILP; i.e. the simple configurations in ${\cal B}_2$ with value
$x_i^{(\ell)} > 0$. The total number of simple configurations $Q_s^{(\ell)}$ in one bin block can be bounded, using observation $7$, by $2^{O(\log^2(1/\delta))}$. Therefore, the total number of simple configurations in ${\cal B}_2$ is $\sum_{\ell=1}^L |Q_s^{(\ell)}| \le L \cdot 2^{O(\log^2(1/\delta))} = 2^{O(\log^2(1/\delta))}$.
This implies that the number of choices for the support of $x$ is at most

$$ {\sum_\ell |Q_s^{(\ell)}| \choose O(1/\delta \log^2(1/\delta))}
  = { 2^{O(\log^2(1/\delta))} \choose O(1/\delta \log^2(1/\delta))}
  = 2^{O(1/\delta \log^4(1/\delta))}. $$

For each choice we solve a reduced MILP with $d= O(1/\delta \log^2(1/\delta))$ integral variables (step 4). The total size $s$ of the MILP can be bounded by $s \le poly(n,1/\delta) + n \log(n) 2^{O(1/\delta \log(1/\delta))}$. Using the algorithm by Kannan with runnning time $d^{O(d)} poly(s)$ for an MILP with $d$ variables and size $s$, we obtain a running time to solve one MILP in time $2^{O(1/\delta \log^3(1/\delta))} poly(n)$; using $poly(s) \le poly(n) 2^{O(1/\delta \log(1/\delta))}$. Running over all vectors $v,w$ and all guesses for the simple configurations, we obtain a running time of $2^{O(1/\delta \log^4(1/\delta))} + poly(n)$. 
\end{proof}

The rounding of the fractional variables in the MILP solutions and the packing of the items accordingly works as in \cite{jansen_eptas_2010} and can be done in time $2^{O(1/\delta \log(1/\delta))} poly(n)$. Therefore, the overall running time of the entire algorithm can be bounded by $2^{O(1/\delta \log^4(1/\delta))} + poly(n) = 2^{O(1/\e \log^4(1/\e))} + poly(n) $; using that $\delta \in O(\e)$.

In order to calculate the length of the computed schedule and to specify $\delta$, we use the following result:

\begin{lemma} \cite{jansen_eptas_2010}
If there is a feasible solution of an MILP instance with bin capacities
$\bar{c}(\ell)$ for blocks $B_\ell \in {\cal B}_2 \cup {\cal B}_3$ and
capacities $\bar{c}_i$ for the $K$ largest bins in ${\cal B}_1$, then the entire job set ${\cal J}$ can be packed into bins with capacities $\bar{c}(\ell)(1+2\delta)^2$ for blocks $B_\ell \in {\cal B}_2 \cup {\cal B}_3$ and enlarged capacities $\bar{c}_i (1+3\delta)^2$ for the first $K$ bins.
\end{lemma}

Note that the result above is constructive, too. This means that there is also an algorithm that computes a corresponding packing \cite{jansen_eptas_2010}. Using $\bar{c}_i \le c_i (1+\delta)^3$ and $T \le (1+\delta) OPT$ and the lemma above, we can bound the schedule length. If there is a schedule with length at most $T$ and with corresponding bin sizes $c_i = T s_i$,
then the lemma above implies a packing into bins of size $c_i (1+3\delta)^3 (1+3\delta)^2$ and a corresponding schedule length $\le T (1+\delta)^3 (1+3\delta)^2 \le OPT (1+\delta)^4 (1+3\delta)^2 \le OPT (1+16\delta) \le OPT(1+\e)$ for $\delta \le \e/16$ and $\e \le 1$. Using $\delta = \frac{1}{\lceil 16/\e \rceil}$, we obtain $\delta \le \e/16$, $\delta \ge \e/17$, and that $1/\delta = \lceil 16/\e \rceil$ is integral. This concludes the proof for Theorem \ref{thm:QCmax}.

\bibliographystyle{plain}
\bibliography{ptas}

\end{document}